\documentclass[11pt,twoside]{amsart}
\usepackage{setspace}
\usepackage{amsmath,hyperref}
\usepackage{amstext}
\usepackage{amssymb}
\usepackage{amsthm} 
\usepackage{mathrsfs}
\usepackage{eurosym}
\usepackage{xcolor}
\usepackage{enumerate}   
\usepackage{enumitem}
\usepackage{tikz}
\usepackage{booktabs}
\usepackage{geometry}
\newtheorem{theorem}{Theorem}[section]
\newtheorem{lemma}[theorem]{Lemma}
\newtheorem{proposition}[theorem]{Proposition}
\theoremstyle{definition}
\newtheorem{defin}[theorem]{Definition}

\newtheorem{assumption}[theorem]{Assumption}
\newtheorem{condition}[theorem]{Condition}
\newtheorem{remark}[theorem]{Remark}

\numberwithin{equation}{section}

\newcommand{\R}{\mathbb{R}}

 \newcommand{\G}{\mathcal{G}}
 \newcommand{\bG}{\bar{G}}
 \newcommand{\wG}{\widetilde{G}}
 \newcommand{\nH}{\mathcal{H}}
 \newcommand{\hH}{\hat{\mathcal{H}}}

\newcommand{\be}{\begin{equation}}
\newcommand{\ee}{\end{equation}}
\newcommand{\ba}{\begin{aligned}}
\newcommand{\ea}{\end{aligned}}

\title{The geometry of multi-curve interest rate models}

\author[C. Fontana]{Claudio Fontana}
\address{Department of Mathematics ``Tullio Levi - Civita'', University of Padova, Italy.}
\email{fontana@math.unipd.it}

\author[G. Lanaro]{Giacomo Lanaro}
\address{Department of Mathematics ``Tullio Levi - Civita'', University of Padova, Italy.}
\email{glanaro@math.unipd.it}

\author[A. Murgoci]{Agatha Murgoci}
\address{Centrica Energy, Denmark.}
\email{Agatha.murgoci@centrica.com}

\date{\today}

\thanks{
The authors are grateful to the participants of the XXIV Workshop on Quantitative Finance (Gaeta, Italy, 2023) and to two anonymous reviewers for useful comments on an earlier version of this work. The first author gratefully acknowledges financial support from the Europlace Institute of Finance and the University of Padova (research programmes  STARS StG PRISMA and BIRD190200/19).}

\begin{document}

\maketitle

\begin{abstract}
We study the problems of consistency and of the existence of finite-dimensional realizations for multi-curve interest rate models of Heath-Jarrow-Morton type, generalizing the geometric approach developed by T. Bj\"ork and co-authors for the classical single-curve setting. We characterize when a multi-curve interest rate model is consistent with a given parameterized family of forward curves and spreads and when a model can be realized by a finite-dimensional state process. We illustrate the general theory in a number of model classes  and examples, providing explicit constructions of finite-dimensional realizations. Based on these theoretical results, we perform the calibration of a three-curve Hull-White model to market data and analyse the stability of the estimated parameters.
\end{abstract}

\section{Introduction}
\label{section - introduction}


In dynamic models for the term structure of interest rates, two fundamental problems concern the {\em consistency} between a forward rate model $\mathcal{M}$ and a parameterized family $\mathcal{G}$ of forward curves and the {\em existence of finite-dimensional realizations} (FDRs) for the model. More specifically, consistency between $\mathcal{M}$ and $\mathcal{G}$ means that, if the initial term structure belongs to $\mathcal{G}$, then model $\mathcal{M}$ will only generate forward rate curves belonging to $\mathcal{G}$, at least for a strictly positive time. 
The existence of finite-dimensional realizations corresponds to the existence of a finite-dimensional Markov state process driving the evolution of the inherently infinite-dimensional term structure. These problems have been addressed and solved in a remarkable series of works by T. Bj\"ork and co-authors (see \cite{bjork2004geometry} for an overview), exploiting the geometric properties of interest rate models and revealing the deep connections between the two problems. 

Up to now, the geometric theory of interest rate models has remained restricted to the classical single-curve setup, where a single term structure provides a complete description of the interest rate market.
However, starting from the 2008 global financial crisis, the emergence of credit, liquidity and funding risks in interbank transactions has led to {\em multi-curve} interest rate markets. The multi-curve phenomenon refers to the coexistence of multiple term structures, each of them associated to an interest rate benchmark for a specific tenor (i.e., the time length of the underlying loan). While overnight rates (such as the recently introduced SOFR in the US, SONIA in the UK, \euro STR in the Eurozone) can be considered {\em risk-free rates}, other interest rate benchmarks (such as Euribor and Libor rates, or the newly proposed Ameribor rates) are {\em risk-sensitive rates} associated to distinct term structures, with a specific behavior depending on the tenor.\footnote{After the cessation of Libor rates, market participants have expressed the need of risk-sensitive rates embedding credit, funding and liquidity risk components. In the US, this has led to the proposal of BSBY and Ameribor rates. In the Eurozone, Euribor rates represent risk-sensitive rates. Therefore, even after the cessation of Libor rates, global interest rate markets continue to be multi-curve interest rate markets.}
The mathematical analysis of multi-curve interest rate markets is made complex by the fact that all interest rate benchmarks are quoted in the same financial market, thereby introducing a strong dependence among the multiple term structures. Hence, multi-curve interest rate markets cannot be adequately described by a naive juxtaposition of several single-curve interest rate models. 
We refer the reader to \cite{grbac2015interest} for an overview of multi-curve interest rate models and,  closer to the modeling approach adopted in the present work, to \cite{cuchiero2016general,fontana2020term} for general multi-curve models of Heath-Jarrow-Morton (HJM) type.

In this paper, we address the problems of consistency and of the existence of finite-dimensional realizations for multi-curve interest rate models, extending the geometric approach first proposed by T. B\"ork and co-authors. This objective is made easier by the fact that several foundational results of \cite{bjork1999interest,bjork2001existence} are formulated at an abstract level, which facilitates their application beyond the classical single-curve setting.
We work in a Heath-Jarrow-Morton model driven by a multi-dimensional Brownian motion and we adopt the convenient parameterization of  \cite{fontana2020term} of multi-curve interest rate markets in terms of spreads and fictitious zero-coupon bond prices. This parameterization highlights the analogy between multi-curve interest rate markets and foreign currency markets. By exploiting this analogy, we can adapt to our setting the methodology of \cite{slinko2010finite} for characterizing finite-dimensional realizations of a two-economy HJM model.
We study in detail the classes of constant volatility models and constant direction volatility models, providing explicit constructions of finite-dimensional realizations. We also study the possibility of including directly the spread processes in the state process determining the finite-dimensional realizations. Finally, we propose a calibration methodology that computes the parameterized manifold that achieves the best fit to market data.

The problems of consistency and of the existence of finite-dimensional realizations have relevant practical applications. Indeed, as explained in \cite[Section 3.1]{bjork2004geometry}, the consistency problem is related to parameter recalibration. Forward rate curves are typically described through a parameterized family of functions $\mathcal{G}$ (such as the popular Nelson-Siegel family of \cite{nelson1987parsimonious}) and, once a forward rate curve has been obtained, an interest rate model $\mathcal{M}$ can be calibrated to it. On the next day, a new forward rate curve is computed and the model $\mathcal{M}$ recalibrated to it. If consistency holds between model $\mathcal{M}$ and the parameterized family of functions $\mathcal{G}$, then $\mathcal{M}$ generates forward rate curves that belong to $\mathcal{G}$. 
Concerning the existence of finite-dimensional realizations, it has to be noted that Heath-Jarrow-Morton models are in general infinite-dimensional. However, if a finite-dimensional realization can be found, then the model becomes significantly easier to handle and can be described by an underlying Markovian factor process. 

We close this introduction by briefly discussing some related literature.
The problem of consistency was first addressed in \cite{bjork1999interest}, while the existence of FDRs was studied in \cite{bjork1999minimal}, \cite{bjork2001existence} and \cite{bjork2002construction}. The inclusion of stochastic volatility process has been addressed in \cite{bjork2004finite}. These results are based on the interpretation of the realization of a forward rate model as a curve living on suitable Hilbert space. In \cite{filipovic2003existence,filipovic2004geometry}, finite-dimensional realizations are studied in the context of forward rate models living on Fr\'echet spaces. For simplicity of presentation, in this work we shall only consider the case of Hilbert spaces, also because the conditions obtained by \cite{bjork2001existence} continue to hold in that more general setting.
Geometric properties related to the problem of consistency and the existence of FDRs for L\'evy models are studied in \cite{filipovic2008existence,tappe2010alternative,tappe2012existence}.

The paper is structured as follows. In Section \ref{section - modelling framework}, we introduce the main modeling quantities of multi-curve interest rate models and the general mathematical framework of our work. In Section \ref{section - consistency problem}, we address the consistency problem, while Section \ref{section - FDRs} contains the study of finite-dimensional realizations. In Section \ref{section - alternative invariance}, we propose and characterize an alternative notion of invariance. Finally, in Section \ref{section - calibration algorithm}, we describe a calibration algorithm that determines the parameterized manifold that achieves the best fit to market data of a multi-curve interest rate market.

\subsection*{Notation}
We introduce some general notation that is going to be used in the paper:
\begin{itemize}
\item We denote by $A^\top$ the transpose of a matrix $A$ and by $v \cdot w$ the scalar product between two vectors $v$ and $w$ in $\R^n$. The Euclidean norm of a vector $v$ is denoted by $\|v\|$.
\item For a differentiable function $f:\R_+\times\R_+\rightarrow\R^n$ we introduce the functionals
\[
\mathbf{F}f(t, x)	:= \frac{\partial}{\partial x} f(t,x),\quad 
\mathbf{H}f(t,x)	:= \int_0^x f(t,u) du,\quad
\mathbf{B}f(t,x)	:= f(t,0).
\]
\item For a Fr\'echet-differentiable function $f:\mathcal{H}_1\rightarrow \mathcal{H}_2$ we denote by $\partial_{\hat{r}} f(\hat{r})$ its Fr\'echet derivative at $\hat{r}\in\mathcal{H}_1$. 
\item We denote by $\mathbb{I}$ the identity map on a vector space $\hat{\mathcal{H}}$. Moreover, if $\mathcal{H} = \R^k$, for $k\in\mathbb{N}$, we denote by $\mathbb{I}_k$ the identity map.
\end{itemize}


\section{The modeling framework} 
\label{section - modelling framework}

In this section, we describe the general framework of multi-curve interest rate models. In Section \ref{section - market setup}, we introduce the generic types of interest rates considered in our analysis and the key modeling quantities. The mathematical setup of our work is then described in Section \ref{section - mathematical framework}.

\subsection{Interest rates and spreads}
\label{section - market setup}

We consider a generic interest rate market with a num\'eraire given by the savings account associated to a {\em risk-free rate} (RFR). As mentioned in the Introduction, the RFR can represent one of the recently introduced overnight interest rate benchmarks. As usual, we parametrize the RFR term structure by means of zero-coupon bond (ZCB) prices, denoting by $B^0_t(T)$ the price at time $t$ of a ZCB with maturity $T$, for all $0\leq t\leq T<+\infty$. The simply compounded forward RFR for the time interval $[T,T+\delta]$ evaluated at time $t\leq T$ is given by $L^0_t(T,T+\delta):=(B^0_t(T)/B^0_t(T+\delta)-1)/\delta$, for $\delta>0$.

Besides the risk-free rate, we consider {\em risk-sensitive rates} (RSR) that reflect the presence of credit, funding and liquidity risk in interbank transactions. As mentioned in the Introduction, RSRs can play the role of Libor/Euribor rates as well as of the newly proposed credit-sensitive rates (e.g., Ameribor). Since risk-sensitive rates exhibit a distinct behavior depending on their reference tenor (i.e., the length of time of the underlying loan), we consider a family of RSRs associated to a set of tenors $\Delta:=\{\delta_1,\ldots,\delta_m\}$, with $\delta_1<\ldots<\delta_m$, for some $m\in\mathbb{N}$. The simply compounded forward RSR for tenor $\delta_j\in\Delta$ is denoted by $L^j_t(T,T+\delta_j)$, for all $0\leq t\leq T<+\infty$.

The multi-curve setup consists in the coexistence of the RFR together with the family of RSRs. In line with \cite{cuchiero2016general,fontana2020term}, instead of modeling RSRs directly, we consider multiplicative {\em spreads} between RSR and RFR, defined as follows:
\be\label{spread-definition}
S^j_t:=\frac{1+\delta_j L^j_t(t,t+\delta_j)}{1+\delta_j L^{0}_t(t,t+\delta_j)},
\qquad\text{ for }t\geq0 \text{ and }j=1,\ldots,m.
\ee
The spread $S^j_t$ can be regarded as a spot measure at time $t$ of the credit, funding, liquidity risks of the interbank market over a time period of length $\delta_j$. Under typical market conditions, spreads are greater than one and increasing with respect to the tenor's length.

For modeling purposes, we introduce fictitious ZCB prices, defined as follows:
\be	\label{delta bond}
B^j_t(T):= \frac{B^0_t(T+\delta_j)}{B^0_t(t+\delta_j)}
\frac{1+\delta_j L^j_t(T,T+\delta_j)}{1+\delta_j L^j_t(t,t+\delta_j)},
\quad\text{ for }0\leq t\leq T<+\infty\text{ and }j=1,\ldots,m.
\ee
Observe that \eqref{delta bond} ensures the terminal condition $B^j_T(T)=1$, for all $T\geq0$ and $j=1,\ldots,m$. We point out that we do not assume that the fictitious bonds introduced above are traded in the market. Rather, fictitious bonds serve as a particularly convenient parametrization of the term structures associated to risk-sensitive rates (compare with \cite[Section 2]{fontana2020term}).

\begin{remark}[FX analogy]	\label{rem:FX}
The quantities $S^j_t$ and $B^j_t(T)$ admit an interpretation in the context of foreign exchange (FX) markets. Indeed, for each $j=1,\ldots,m$, one can associate to the tenor $\delta_j$ a foreign economy denominated in a specific currency, whose currency risk is representative of the level of credit and liquidity risks implicit in the interbank market for tenor $\delta_j$.
The spread $S^j_t$ can be thought of as the spot exchange rate at time $t$ between the $j$-th foreign economy and the domestic economy, while $B^j_t(T)$ represents the price (in units of the foreign currency) at time $t$ of a ZCB with maturity $T$ of the $j$-th foreign economy. 
According to this interpretation, a swap referencing $L^j_T(T,T+\delta_j)$ can be thought of as an FX forward contract where one unit of the $j$-th foreign currency is delivered against a fixed payment in the domestic currency. Accordingly, it can be shown that the value at time $t$ of the floating leg of a single-period swap referencing $L^j_T(T,T+\delta_j)$ is given by $S^j_tB^j_t(T)$, see \cite[Section 2]{fontana2020term}.
This FX viewpoint on multi-curve interest rate models goes back to the work of \cite{Bianchetti}, has been further discussed in \cite{fontana2020term,MM16,NguyenSeifried15} (see also \cite{cuchiero2016general}\footnote{We point out that the FX analogy presented in \cite[Appendix B]{cuchiero2016general} is based on a different definition of ZCB prices associated to RSR and does not yield a clear interpretation of the corresponding spot exchange rate process.}) and is consistent with the general framework for multiple term structures first formulated by \cite{JarrowTurnbull98}.
In our context, this FX analogy will enables us to study finite-dimensional realizations of multi-curve interest rate models (see Section \ref{section - FDRs}) by relying on and extending the approach of \cite{slinko2010finite}, where the existence of finite-dimensional realizations for two-currency markets has been analyzed.
\end{remark}

\subsection{Term structure dynamics}
\label{section - mathematical framework}

Let $(\Omega, \mathcal{F}, (\mathcal{F}_t)_{t\geq0},\mathbb{Q})$ be a filtered probability space, endowed with a $d$-dimensional Brownian motion $(W_t)_{t\geq0}$ and where $\mathbb{Q}$ is a risk-neutral probability. 
In order to describe the RFR and RSR term structure dynamics, we adopt the Heath-Jarrow-Morton methodology, referring to \cite{fontana2020term} for additional details on the general framework. 

Adopting the Musiela parametrization, we represent risk-free and fictitious ZCB prices as 
\be	\label{eq:bond_prices}
B^j_t(T) = \exp\left(-\int_0^{T-t}r^j_t(x)dx\right),
\qquad\text{ for all }0\leq t\leq T<+\infty\text{ and }j=0,1,\ldots,m.
\ee
For $(t,x)\in\R^2_+$, the quantity $r^0_t(x)$ represents the risk-free instantaneous forward rate at time $t$ for maturity $t+x$. Similarly, for each $j=1,\ldots,m$ and $(t,x)\in\R^2_+$, the quantity $r^j_t(x)$ represents the risk-sensitive instantaneous forward rate at time $t$ for maturity $t+x$ relative to tenor $\delta_j$.
The savings account num\'eraire associated to the RFR is given by $S^0:=\exp(\int_0^{\cdot}r^0_t(0)dt)$.

Risk-free and risk-sensitive instantaneous forward rates are assumed to satisfy
\be	\label{eq:SDE_forward}
dr^j_t(x) = \alpha^j_t(x) dt + \sigma^j_t(x)dW_t,
\qquad\text{ for all }j=0,1,\ldots,m,
\ee
where $\alpha^j:\Omega\times\R_+^2\to\R$ and $\sigma^j:\Omega\times\R_+^2\to\R^d$ are progressively measurable stochastic processes satisfying suitable integrability requirements to ensure the well-posedness of \eqref{eq:bond_prices} and \eqref{eq:SDE_forward}.

The spreads introduced in \eqref{spread-definition} are modelled as exponentials of It\^o processes:
\be	\label{eq:SDE_spread}
S^j_t = \exp(Y^j_t),
\qquad\text{ where }\;
dY^{j}_t=\gamma^{j}_tdt+\beta^{j}_tdW_t,
\ee
for all $j=1,\ldots,m$, where $\gamma^j:\Omega\times\R_+\to\R$ and $\beta^j:\Omega\times\R_+\to\R^d$ are suitable progressively measurable processes ensuring the existence of a unique strong solution to \eqref{eq:SDE_spread}.
We shall refer to the process $Y^j$ as the {\em log-spread} process associated to tenor $\delta_j$, for $j=1,\ldots,m$.

\begin{remark}
Equation \eqref{eq:SDE_spread} represents a generic modeling framework for multiplicative spreads. As mentioned above, spreads are typically greater than one and ordered with respect to the tenor. These features can be ensured by a suitable specification of the processes $Y^j$, for $j=1,\ldots,m$. 
Depending on the model structure, this can also ensure the positivity of basis swap spreads.
\end{remark}

In the classical single-curve setting, the HJM drift condition implies that the drift term $\alpha^0$ in \eqref{eq:SDE_forward}  is determined by the volatility $\sigma^0$ (see, e.g., \cite[Proposition 1.1]{bjork2004geometry}). 
In the present multi-curve setup, risk-neutrality of $\mathbb{Q}$ implies that, for each $j=1,\ldots,m$, the drift term $\alpha^j$ in \eqref{eq:SDE_forward} is determined by the volatility $\sigma^j$ as well as by the covariation between $r^j$ and the log-spread process $Y^j$. Moreover, the drift term $\gamma^j$ in \eqref{eq:SDE_spread} turns out to be endogenously determined. This is the content of the following proposition, which follows as a special case of \cite[Theorem 3.7]{fontana2020term}.
For convenience of notation, we set $\beta^0:=0$ and $\gamma^0:=0$ in the following.

\begin{proposition}	\label{prop:drift}
Under a risk-neutral probability measure $\mathbb{Q}$, the following holds:
\begin{align*}
\alpha^j_t(x) &= \mathbf{F}r^j_t(x) + \sigma^j_t(x)\cdot\mathbf{H}\sigma^j_t(x) - \beta^j_t\cdot\sigma^j_t(x),\\
\gamma^j_t &= \mathbf{B}r^0_t-\mathbf{B}r^j_t - \frac{1}{2}\|\beta^j_t\|^2,
\end{align*}
for every $(t,x)\in\R^2_+$ and $j=0,1,\ldots,m$.
\end{proposition}

\begin{remark}
(i) The drift conditions stated in Proposition \ref{prop:drift} are equivalent to the local martingale property under $\mathbb{Q}$ of the processes $B^0(T)/S^0$ and $S^jB^j(T)/S^0$, for all $T>0$ and $j=1,\ldots,m$. This property is taken as the defining property of a risk-neutral probability. As clarified in \cite{fontana2020term}, this suffices to ensure absence of arbitrage in the financial market composed by all risk-free ZCBs and  single-period swaps\footnote{Single-period swaps represent the basic building blocks of interest rate derivatives having RSR as underlying.} referencing the risk-sensitive rates $L^j_T(T,T+\delta_j)$, for all $T>0$ and $j=1,\ldots,m$. 
It is important to note that in our setup RSR are benchmark rates that serve to define payoffs of interest rate derivatives
 and do not represent actual borrowing/lending rates in the money market. Indeed, allowing for borrowing/lending at RSR would require the explicit modeling of refinancing risk (stemming from credit, liquidity and funding risks, see \cite{BMSS23,FPR23}), which is beyond the scope of the present work.

(ii) The drift conditions stated in Proposition \ref{prop:drift} are analogous to those appearing in \cite[Section 2]{slinko2010finite} in the context of a two-economy HJM model (compare also with Remark \ref{rem:FX} above).
\end{remark}

As a consequence of \eqref{eq:SDE_forward}-\eqref{eq:SDE_spread}, our modeling framework is described by $2m+1$ stochastic differential equations (SDEs), since we consider $m+1$ forward rates and $m$ spreads. In addition, the SDE \eqref{eq:SDE_forward} depends on the time-to-maturity variable $x\in\R_+$. In line with \cite{bjork2004geometry}, for each $j=0,1,\ldots,m$, we view \eqref{eq:SDE_forward} as an SDE taking values in a function space $\mathcal{H}\subseteq\mathcal{C}^{\infty}(\mathbb{R}_+,\mathbb{R})$.
More specifically, we assume that each forward rate process $r^j$ takes value in the space
\[
\mathcal{H}:=\bigg{\{}r\in\mathcal{C}^{\infty}(\R_+,\R)\text{ such that }||r||^2_{\gamma}:=\sum_{n=0}^{+\infty}2^{-n}\int_0^{+\infty}\Bigl{(}\frac{\partial^n}{\partial x^n}r(x)\Bigr{)}^2e^{-\gamma x}dx<+\infty\bigg{\}},
\]
with $\gamma>0$.  By \cite[Proposition 4.2]{bjork2001existence}, the space $(\mathcal{H},||\cdot||_{\gamma})$ is an Hilbert space for every $\gamma>0$. 

We denote by $\hat{r}:=(r^0,r^1,\ldots,r^m,Y^1,\ldots,Y^m)$ the solution to \eqref{eq:SDE_forward}-\eqref{eq:SDE_spread}, with drift terms determined as in Proposition \ref{prop:drift}. As explained below, existence of a unique solution to \eqref{eq:SDE_forward}-\eqref{eq:SDE_spread} is guaranteed under Assumption \ref{smoothness coefficients}.
The process $\hat{r}$ takes values on the space $\hat{\mathcal{H}}:=\mathcal{H}^{m+1}\times\mathbb{R}^m$, where $\mathcal{H}^{m+1}$ denotes the Cartesian product of $m+1$ copies of $\mathcal{H}$. We shall also adopt the notation $\hat{r} = (r,Y)\in\hat{\mathcal{H}}$, where $r\in\mathcal{H}^{m+1}$ represent the forward rates and $Y\in\R^m$ the log-spreads.

We introduce some technical assumptions on the volatility coefficients of \eqref{eq:SDE_forward} and \eqref{eq:SDE_spread}. In the following, Assumption \ref{smoothness coefficients} is always assumed to be satisfied without further mention.

\begin{assumption}
\label{smoothness coefficients}
(i) For all $j=0,1,\ldots,m$ and $(t,x)\in\R^2_+$, it holds that
\begin{equation}	\label{volatilities}
\sigma^j_t(x) = \sigma^j(\hat{r}_t)(x)
\qquad\text{ and }\qquad
\beta^j_t = \beta^j(\hat{r}_t),  
\end{equation}
where $\sigma^j:\hat{\mathcal{H}}\times\R_+\to\R^d$ and $\beta^j:\hat{\mathcal{H}}\to\R^d$ are deterministic functions.

(ii) For each $j=0,1,\ldots,m$, the functions $\sigma^j$ and $\beta^h$ appearing in \eqref{volatilities} are smooth functions in the Fr\'echet sense (i.e., they admit continuous $n$-th order Fr\'echet derivative, for every $n\in\mathbb{N}$). In addition, the following functions are smooth in the Fr\'echet sense:
\[
\sigma^j(\hat{r})\cdot\textbf{H}\sigma^j(\hat{r})-\frac{1}{2}\partial_{\hat{r}}\sigma^j(\hat{r})\hat{\sigma}(\hat{r})-\sigma^j(\hat{r})\cdot\beta^{j}(\hat{r}),
\qquad \text{ for all }j = 0,1,\ldots,m.
\]
\end{assumption}

Under Assumption \ref{smoothness coefficients}, we can compactly write as follows the dynamics of the process $\hat{r}$:
\be	\label{eq:ito}
d\hat{r}_t=\mu(\hat{r}_t)dt+\hat{\sigma}(\hat{r}_t)dW_t,
\ee
where $\hat{\sigma}(\hat{r}_t):=(\sigma^0(\hat{r}_t),\sigma^1(\hat{r}_t), \ldots,\sigma^m(\hat{r}_t), \beta^1(\hat{r}_t), \ldots, \beta^m(\hat{r}_t))^{\top}\in\hat{\mathcal{H}}^d$ and $\mu(\hat{r}_t)$ is an element of $\hat{\mathcal{H}}$ collecting all drift terms of equations \eqref{eq:SDE_forward}-\eqref{eq:SDE_spread}. 
It is a standard fact that Assumption \ref{smoothness coefficients} ensures the local existence of a unique strong solution to \eqref{eq:ito} in the function space $\hat{\mathcal{H}}$ (see, e.g., \cite[Section 2]{slinko2010finite}). For completeness of presentation, we give the following proposition.

\begin{proposition}	\label{prop:existence_solution}
Under Assumption \ref{smoothness coefficients}, there exists an a.s. strictly positive stopping time $\tau$ such that there exists a unique strong solution $\hat{r}^{\tau}=(\hat{r}_{\tau\wedge t})_{t\geq0}$  to \eqref{eq:ito} in $\hat{\mathcal{H}}$, for every $\hat{r}_0\in\hat{\mathcal{H}}$. 
\end{proposition}
\begin{proof}
Proposition \ref{prop:drift} together with Assumption \eqref{smoothness coefficients} implies that $\mu$ is a smooth vector field in $\hat{\mathcal{H}}$. Hence, again by Assumption \ref{smoothness coefficients}, the coefficients of \eqref{eq:ito} are smooth in the Fr\'echet sense and, therefore, locally Lipschitz continuous. As a consequence of \cite[Theorems 7.2 and 6.5]{DPZ} (see also \cite[Corollary 2.4.1]{Filipovic2001}), there exists a unique continuous weak solution to \eqref{eq:ito} up to a stopping time $\tau>0$ a.s. 
By \cite[Proposition 4.2]{bjork2001existence}, the operator  $\mathbf{F}$ is bounded in $\hat{\mathcal{H}}$. In view of \cite[Proposition F.0.4]{PR07}, this allows to conclude that the unique weak solution to \eqref{eq:ito} is a strong solution.
\end{proof}


In the following, it turns out to be convenient to rewrite \eqref{eq:ito} in terms of the Stratonovich integral (see, e.g., \cite[Definition 3.3.13]{karatzas2012brownian}), the reason being that in Stratonovich calculus It\^o's formula takes the form of the usual chain rule of ordinary calculus.
Denoting by $\int X_s \circ dY_s$ the Stratonovich integral of a semimartingale $X$ with respect to a semimartingale $Y$, we then have
\begin{equation}
\label{forward-rate system}
d\hat{r}_t =\hat{\mu}(\hat{r}_t)dt+\hat{\sigma}(\hat{r}_t)\circ dW_t, 
\qquad\text{ with } 
\hat{\mu}(\hat{r}_t):=  \mu(\hat{r}_t)-\frac{1}{2}\partial_{\hat{r}}\hat{\sigma}(\hat{r}_t)\hat{\sigma}(\hat{r}_t), 
\end{equation}
where, for every $\hat{r}\in\hat{\mathcal{H}}$, the drift term $\hat{\mu}(\hat{r})$ can be explicitly written as follows:
\begin{equation}
\label{stratonovich coefficients}
\hat{\mu}(\hat{r})=
\begin{pmatrix}
\textbf{F}r^0+\sigma^0(\hat{r})\cdot\textbf{H}\sigma^0(\hat{r})\\ 
\textbf{F}r^1+\sigma^1(\hat{r})\cdot\textbf{H}\sigma^1(\hat{r})-\sigma^1(\hat{r})\cdot\beta^{1}(\hat{r})\\ 
\vdots\\ 
\textbf{F}r^m+\sigma^m(\hat{r})\cdot\textbf{H}\sigma^m(\hat{r})-\sigma^m(\hat{r})\cdot\beta^{m}(\hat{r})\\ 
\textbf{B}r^0-\textbf{B}r^1-\frac{1}{2}||\beta^1(\hat{r})||^2\\ 
\vdots\\ 
\textbf{B}r^0-\textbf{B}r^m-\frac{1}{2}||\beta^m(\hat{r})||^2
\end{pmatrix}-\frac{1}{2}
\partial_{\hat{r}}\hat{\sigma}(\hat{r})
\begin{pmatrix}
\sigma^0(\hat{r})\\ 
\sigma^1(\hat{r})\\ 
\vdots\\ 
\sigma^m(\hat{r})\\ 
\beta^1(\hat{r})\\ 
\vdots\\ 
\beta^m(\hat{r})
\end{pmatrix},
\end{equation}
\[
\partial_{\hat{r}}\hat{\sigma}(\hat{r})=
\begin{pmatrix}
\partial_{r^0} \sigma^0(\hat{r})		&	\cdots	&	\partial_{r^m} \sigma^0(\hat{r})		&	\partial_{Y^1} \sigma^0(\hat{r})		&	\cdots 	&	\partial_{Y^m} \sigma^0(\hat{r})	\\ 
\vdots						&			&	\vdots						&	\vdots						&			& 	\vdots						\\ 
\partial_{r^0} \sigma^m(\hat{r}	)	&	\cdots	&	\partial_{r^m} \sigma^m(\hat{r})		&	\partial_{Y^1} \sigma^m(\hat{r})		&	\cdots	&	\partial_{Y^m} \sigma^m(\hat{r})	\\ 
\partial_{r^0} \beta^1(\hat{r})				&	\cdots	&	\partial_{r^m} \beta^1(\hat{r})			&	\partial_{Y^1} \beta^1(\hat{r})		&	\cdots	&	 \partial_{Y^m} \beta^1(\hat{r})	\\ 
\vdots						&			&	\vdots						&	\vdots						&			& 	\vdots						\\ 
\partial_{r^0} \beta^m(\hat{r})			&	\cdots	&	\partial_{r^m} \beta^m(\hat{r})			&	\partial_{Y^1} \beta^m(\hat{r})		&	\cdots	&	 \partial_{Y^m} \beta^m(\hat{r})	\\ 
\end{pmatrix}.
\]

\begin{remark}[On the relation to stochastic volatility models]
For multi-curve interest rate models, the problems of consistency and existence of FDRs cannot be addressed by adapting the techniques used in \cite{bjork2004finite} for stochastic volatility models, replacing the stochastic volatility process with the log-spread processes. Indeed, in \cite{bjork2004finite} the stochastic volatility dynamics do not depend on the forward curves. 
On the contrary, in a general multi-curve interest rate model, the dynamics of the log-spread processes cannot be separated from the forward curves, also as a consequence of the drift conditions stated in Proposition \ref{prop:drift}.
\end{remark}

\section{The consistency problem}
\label{section - consistency problem}

In this section, we study the consistency between a multi-curve interest rate model and a parameterized family of forward curves\footnote{In this section, with some abuse of terminology, we shall refer to $G$ as a parameterized family of forward curves even if, strictly speaking, only the first $m+1$ components of $G$ represent forward curves, the last $m$ components (corresponding to the log-spreads) being real-valued. An alternative approach will be presented in Section \ref{section - alternative invariance}.}. In Section \ref{sec:gen_consistency}, we provide a general characterization of consistency, extending to the multi-curve setting the geometric approach first introduced by \cite{bjork1999interest}. In Section \ref{subsection - HL NS}, we provide a detailed analysis of an example,  addressing the consistency between Hull-White multi-curve models and a modified Nelson-Siegel family of forward curves.

\subsection{Characterization of consistency}
\label{sec:gen_consistency}

We consider a mapping $G$ defined on an open subset $\mathcal{Z}\subseteq\R^n$, for some $n\in\mathbb{N}$. The mapping $G$ determines a manifold $\mathcal{G}\subseteq\hat{\mathcal{H}}$ defined by the image $\text{Im}[G]$. More specifically, in this section we shall work under the following assumption.

\begin{assumption}
\label{assumption - immersion condition} 
The mapping $G:\mathcal{Z}\to\hat{\mathcal{H}}$ is an injective function such that $\partial_zG:\mathbb{R}^n\to\hat{\mathcal{H}}$ is injective, for all $z\in\mathcal{Z}$. As a consequence, $\mathcal{G}:=\text{Im}[G]=\{G(z):z\in\mathcal{Z}\}$ is a submanifold of $\hat{\mathcal{H}}$.
\end{assumption}

For each parameter value $z\in\mathcal{Z}$, the mapping $G$ produces a curve $G(z)\in\hat{\mathcal{H}}$. The value of this curve at the point $x\in\R_+$ is denoted by $G(z,x)$, so that $G$ can also be viewed as a mapping $G:\mathcal{Z}\times\R_+\to\R^{2m+1}$. The mapping $G$ formalizes the idea of a finitely parameterized family of forward curves. The manifold $\mathcal{G}$ represents the set of all curves that can be generated by $G$.

In the following, we denote by $\mathcal{M}$ a multi-curve interest rate model as defined in Section \ref{section - mathematical framework}. 
We recall that, as a consequence of Proposition \ref{prop:drift}, a multi-curve interest rate model $\mathcal{M}$ is entirely determined by the volatilities $\sigma^j$ and $\beta^j$, for $j=0,1,\ldots,m$, which satisfy Assumption \ref{smoothness coefficients}.
A model $\mathcal{M}$ and a submanifold $\mathcal{G}$ are said to be {\em consistent} if model $\mathcal{M}$ generates forward curves that belong to $\mathcal{G}$, at least for a strictly positive time interval. The notion of consistency is precisely defined through the following concept of local invariance (see \cite[Definition 1.1]{bjork1999interest}).

\begin{defin} \label{invariance}
A manifold $\mathcal{G}$ is said to be {\em locally invariant} under the action of $\hat{r}$ if, for each point $(s,\hat{r}_s)\in\R_+\times\mathcal{G}$, there exists a stopping time $\tau(s,\hat{r}_s)$ such that almost surely
\[
\tau(s,\hat{r}_s)>s
\quad\text{ and }\quad
\hat{r}_t\in\mathcal{G},
\text{ for each }t\in[s,\tau(s,\hat{r}_s)).
\]
If $\tau(s,\hat{r}_s)=+\infty$ a.s. for all $(s,\hat{r}_s)\in\R_+\times\mathcal{G}$, the manifold $\mathcal{G}$ is said to be {\em globally invariant}.
\end{defin}

In the following, we give a necessary and sufficient condition for local invariance (and invariance will always be meant in a local sense). To this effect, following the approach of \cite{bjork1999interest}, we exploit the equivalence between the notion of invariance and the notion of $\hat{r}-invariance$.

\begin{defin}	\label{r invariance def}
A parameterized family $G$ is said to be (locally) {\em$\hat{r}$-invariant} under the action of $\hat{r}$ if, for every $\hat{r}_0\in\mathcal{G}$, there exists an a.s. strictly positive stopping time $\tau(\hat{r}_0)$ and a stochastic process $(Z_t)_{t\geq0}$, taking values in $\mathcal{Z}$ and with Stratonovich dynamics 
\begin{equation}\label{eq:Stratonovich_Z}
dZ_t=a(Z_t)dt+b(Z_t)\circ dW_t,
\end{equation}
such that for all $t\in[0,\tau(\hat{r_0}))$ it holds that
\[
r_t(x)=G(Z_t,x)\text{ a.s.},
\qquad\text{for all }x\in\R_+.
\]
\end{defin}

As in the classical single-curve setup considered in \cite{bjork1999interest} and \cite{bjork2004geometry}, the process $(Z_t)_{t\geq0}$ represents an underlying factor process. The notion of $\hat{r}$-invariance is therefore equivalent to the existence of an underlying finite-dimensional factor model. This concept will turn out to be intimately related to the existence of finite-dimensional realizations (see Section \ref{section - FDRs} below). 

Under Assumptions \ref{smoothness coefficients} and \ref{assumption - immersion condition}, it can be easily shown that $\hat{r}$-invariance is equivalent to invariance. 
By relying on this fact, a straightforward adaptation of the proof of \cite[Theorem 4.1]{bjork1999interest} leads to the following characterization of invariance in terms of the vector fields $\hat{\mu}$ and $\hat{\sigma}$.

\begin{theorem} \label{invariance-thm}
The manifold $\mathcal{G}$ is invariant under the action of $\hat{r}$ if and only if the following conditions hold for all $z\in\mathcal{Z}$:
\begin{equation}
\label{invariance condition}
\begin{split}
\hat{\mu}(G(z))		&\in \mathrm{Im}[G_z(z)] = T_{G(z)}\mathcal{G},\\ 
\hat{\sigma}_i(G(z))	&\in \mathrm{Im}[G_z(z)] = T_{G(z)}\mathcal{G},
\quad \text{ for all } i=1,\ldots,d,
\end{split}
\end{equation}
where $T_{G(z)}\G$ denotes the tangent space of $\G$ at the point $G(z)$.
\end{theorem}

Condition \eqref{invariance condition} is equivalent to require that the distribution generated by $\hat{\mu}$ and $\hat{\sigma}$ (i.e., the subspace of the tangent bundle of $\hat{\mathcal{H}}$ generated by $\hat{\mu}$ and $\hat{\sigma_i}$, for $i = 1,\ldots,d$) is a subset of $T\mathcal{G}$, where $T\mathcal{G}$ denotes the tangent bundle of $\mathcal{G}$. 

Differently from the single-curve case, in the present multi-curve setup the study of the consistency between a model $\mathcal{M}$ and a manifold $\mathcal{G}$ depends on the relations among the components of the function $G$, representing forward rates and log-spreads associated to different tenors. This is illustrated in the next section in the case of a Hull-White multi-curve model.

\begin{remark}	\label{rem:joint_spreads}
Even if the log-spread processes are inherently finite-dimensional, it is in general not possible to consider the consistency problem for the forward rate components alone, separately from the log-spread processes. This is due to the fact that the forward rate dynamics can depend on the log-spreads. For this reason, we consider a parameterized family $G$ providing a joint representation of forward rates and log-spreads.
In Section \ref{section - alternative invariance} we will present an alternative approach that allows including the log-spreads explicitly in the state process. 
\end{remark}

\subsection{Example: Hull-White model and modified Nelson-Siegel family}	\label{subsection - HL NS}

For simplicity of presentation, let us consider the case $d=1$ (the results of this section can be extended in a straightforward way to the case of a $d$-dimensional Brownian motion, see \cite[Section 2.3.8]{lanaro2019geometry}). We suppose that, for each $j=0,1,\ldots,m$, the forward rate equation \eqref{eq:SDE_forward} takes the form
\be
\label{Hull White - Musiela parameterization}
dr^j_t(x)=\Big(\textbf{F}r^j_t(x) + \frac{(\sigma^j)^2}{a^j}e^{-a^jx}(1-e^{-a^jx})\Big)dt+\sigma^j e^{-a^jx}dW_t.
\ee
Since the Hull-White forward rate equation has constant volatility, we naturally assume that also the volatilities of the log-spread processes are constant. Hence, the volatility of our Hull-White multi-curve model $\mathcal{M}$ is given by $\hat{\sigma}(x) = (\sigma^0e^{-a^0x},\sigma^1e^{-a^1x}, \ldots, \sigma^me^{-a^mx}, \beta^1, \ldots,\beta^m)$.

We aim at determining a manifold $\G\subseteq\hat{\mathcal{H}}$ such that the conditions of Theorem \ref{invariance-thm} are satisfied. Already in the single-curve setting, it is well-known that the Nelson-Siegel family is inconsistent with the Hull-White model (see \cite[Proposition 5.1]{bjork1999interest}). Therefore, we consider the following {\em modified Nelson-Siegel family}:
\be	\label{Gj augmented}
G_j({z}^j,x) :=z^j_1+z^j_2e^{-a^jx}+z^j_3xe^{-a^jx}+z^j_4e^{-2a^jx},
\ee
denoting $z^j=(z^j_1,z^j_2,z^j_3,z^j_4)\in\R^4$, for each $j=0,1,\ldots,m$.
We then introduce the function
\begin{equation}
\label{structure of G}
G:=(G_0,G_1,\ldots,G_m,G_{m+1},\ldots,G_{2m}),
\end{equation}
where the elements $G_j$, for $j=0,1,\ldots,m$, are given as in \eqref{Gj augmented} and $G_{m+j}$, for $j = 1,\ldots,m$, are suitable real-valued functions to be determined later, corresponding to the log-spread components of the multi-curve model $\mathcal{M}$.
For every $j=0,1,\ldots,m$, letting
\begin{equation}
\label{etaj}
\begin{aligned}
\eta^j(z^j) &:= 
\biggl( 0,	 -a^jz^j_2 + z^j_3 + \frac{(\sigma^j)^2}{a^j}-\beta^j\sigma^j, - a^j z^j_3,  - 2a^j z^j_4 - \frac{(\sigma^j)^2}{a^j}\biggr),\\
\xi^j(z^j) &:= (0,\sigma^j,0,0),
\end{aligned}
\end{equation}
we can easily check that conditions $\partial_{z^j}G(z^j,x)\eta^j(z^j) = \mu^j(x)$ and $\partial_{z^j}G(z^j,x)\xi^j(z^j) = \sigma^je^{-a^jx}$ are satisfied, with $\mu^j(x)$ denoting the drift term of \eqref{Hull White - Musiela parameterization}.
Therefore, the vector fields
\be	\label{eta-xi}
\eta	:= ( \eta^0,\ldots, \eta^m): \R^{4(m+1)} \to \R^{4(m+1)}
\quad\text{ and }\quad
\xi	:= ( \xi^0,\ldots, \xi^m): \R^{4(m+1)} \to \R^{4(m+1)}
\ee
satisfy the first $m+1$ components of the invariance conditions \eqref{invariance condition}. This ensures consistency between the forward rate components of the model and the parameterized manifold $\mathcal{G}=\mathrm{Im}[G]$.

However, we still have to consider the log-spread processes, corresponding to the last $m$ components of the function $G$ in \eqref{structure of G}. To this effect, we can follow two alternative approaches:
\begin{enumerate}
\item[(i)] Exploiting the fact that the log-spread processes are inherently finite-dimensional, we can enlarge the parameter space $\R^{4(m+1)}$ by introducing $m$ additional variables corresponding to the log-spread processes themselves. We therefore consider the enlarged state space $\R^{5m+4}$ and define $G_{m+j}(u^j):=u^j$, for every $u^j\in\R$ and each $j=1,\ldots,m$. With this approach, the conditions of Theorem \ref{invariance-thm} are always satisfied. Indeed, to verify that \eqref{invariance condition} holds, it suffices to take $\xi^{m+j}(z):=\beta^j$ and 
\begin{align*}
\eta^{m+j}(z):= \hat{\mu}^{m+j}(G(z)) &= \mathbf{B}G_0(z^0) - \mathbf{B}G_j(z^j) - \frac{1}{2} (\beta^j)^2	\\
&= z^0_1 + z^0_2 + z^0_4 -  ( z^j_1 + z^j_2 + z^j_4) - \frac{1}{2} (\beta^j)^2,
\end{align*}
for all $z\in\R^{5m+4}$ and $j=1,\ldots,m$.
In this way, we obtain that the parameterized family
\begin{equation}
\label{parameterized family HW - strategy 1}
G({z}):= \bigl(G_0(z^0),G_1(z^1),\ldots, G_m(z^m), u^1,\ldots, u^m\bigr),
\end{equation} 
defined for any ${z}= (z^0,z^1,\ldots, z^m, u^1,\ldots,u^m) \in \R^{5m+4}$, generates a manifold $\mathcal{G}$ that is consistent with the Hull-White multi-curve model under analysis.
\item[(ii)] Instead of enlarging the parameter space $\R^{4(m+1)}$, we can look for conditions that ensure that the last $m$ components of conditions \eqref{invariance condition}, corresponding to the log-spread processes, are automatically satisfied by the vector fields $\eta$ and $\xi$ introduced in \eqref{etaj}-\eqref{eta-xi}, thereby determining implicitly the components $(G_{m+1},\ldots,G_{2m})$ of the function $G$ in \eqref{structure of G}.
\end{enumerate}

Approach (ii) requires the validity of some internal relations among the volatilities of the forward rates and of the log-spread processes, as clarified by the next result.

\begin{proposition}
\label{result HW NS}
Let the manifold $\mathcal{G}$ be given by the image of the function $G:\R^{4(m+1)}\to\hat{\mathcal{H}}$, with $G_j$ given as in \eqref{Gj augmented}, for each $j=0,1,\ldots,m$, and $G_{m+j}$ defined as follows:
\begin{equation}
\label{spread parameterized function - strategy 2}
\begin{aligned}
G_{m+j}(z)	&:= \frac{1}{a^0}\bigg(-z^0_2+\Bigl{(}-z^0_1-\frac{(\sigma^0)^2}{2(a^0)^2}+\frac{1}{2}(\beta^j)^2\Bigr{)}\log{z^0_3} -\frac{z^0_3}{a^0}-\frac{1}{2}z^0_4\bigg)+\\ 
			&\quad +\frac{1}{a^j}\bigg(z^j_2+\Bigl{(}z^j_1+\frac{(\sigma^j)^2}{2(a^j)^2} -\frac{\beta^j\sigma^j}{a^j}\Bigr{)}\log{z^j_3}+\frac{z^j_3}{a^j}+\frac{1}{2}z^j_4\bigg),
\end{aligned}
\end{equation}
for each $j=1,\ldots,m$. If $\beta^j=\sigma^j/a^j-\sigma^0/a^0$, for all $j=1,\ldots,m$, then the manifold $\mathcal{G}$ is consistent with the Hull-White multi-curve model considered in this section.
\end{proposition}
\begin{proof}
By Theorem \ref{invariance-thm}, we must verify that the following two conditions
\begin{align} 
\label{consistency on mu - strategy 2}
\hat{\mu}^{m+j}(G(z))			&= \mathbf{B}G_0(z^0) - \mathbf{B}G_j(z^j) - \frac{1}{2} (\beta^j)^2 = \partial_{z}G_{m+j}(z) \eta(z),\\ 
\label{consistency on sigma - strategy 2}
\hat{\sigma}^{m+j}(G(z))			& = \beta^j  = \partial_{z}G_{m+j}(z) \xi(z),
\end{align}
are satisfied by $\eta$ and $\xi$ given in \eqref{etaj}-\eqref{eta-xi}, for all $j=1,\ldots,m$. Observe first that
\[
\mathbf{B}G_0(z^0) - \mathbf{B}G_j(z^j) - \frac{1}{2} (\beta^j)^2  = z^0_1 + z^0_2 + z^0_4 - z^j_1 - z^j_2- z^j_4 - \frac{1}{2}(\beta^j)^2.
\]
By relying on this identity, differentiating \eqref{spread parameterized function - strategy 2} and making use of \eqref{etaj}, it can be checked that condition \eqref{consistency on mu - strategy 2} holds (we refer to \cite[Section 2.3.5]{lanaro2019geometry} for the detailed computations). Similarly, to check that condition \eqref{consistency on sigma - strategy 2} holds, we differentiate \eqref{spread parameterized function - strategy 2} and make use of the specification $\xi^j=(0,\sigma^j,0,0)$, together with the hypothesis that $\beta^j=\sigma^j/a^j-\sigma^0/a^0$, for all $j=1,\ldots,m$. 
\end{proof}

\begin{remark}
\label{parsimonious strategy}
Approach (i) leads to a parameterized family whose domain is $\R^{5m+4}$, without any further requirement on the model. Approach (ii) leads to a parameterized family whose domain lies in $\R^{4(m+1)}$ and is therefore more parsimonious, but requires the validity of a specific relation. 
It can be shown that the condition $\beta^j=\sigma^j/a^j-\sigma^0/a^0$ implies that the log-spread process $Y^j$ is an affine transformation of the spot rates $(\mathbf{B}r^0,\mathbf{B}r^j)$. If this is the case, a parameterized family that is consistent with the forward rates $(r^0,r^j)$ automatically determines a parameterized family for the log-spread process $Y^j$. This explains why it is not necessary to enlarge the parameter space if $\beta^j=\sigma^j/a^j-\sigma^0/a^0$, for all $j=1,\ldots,m$.
\end{remark}

\section{Finite-dimensional realizations}
\label{section - FDRs}

In this section, we study the existence and the construction of finite-dimensional realizations (FDRs) for a multi-curve interest rate model $\mathcal{M}$ as described in Section \ref{section - mathematical framework}. In Section \ref{subsection - the strategy}, we characterize the existence of FDRs by relying on the geometric approach of \cite{bjork2001existence} and we outline a general procedure for the construction of FDRs. Section \ref{subsection - example constant volatility} illustrates the methodology in the simple case of constant volatility models, while Section \ref{subsection - example constant direction volatility} considers the more complex case of constant direction volatility models. In Section \ref{sec:benchmark}, we show that the state variables of an FDR can be chosen as an arbitrary family of forward rates and log-spreads.

\subsection{Existence and general construction of FDRs} 
\label{subsection - the strategy}

We start by defining the notion of an $n$-dimensional realization for a multi-curve interest rate model $\mathcal{M}$ (see \cite[Definition 3.1]{bjork2001existence}).

\begin{defin}
\label{fdr def}
Model $\mathcal{M}$ has an {\em $n$-dimensional realization} if, for each $\hat{r}_0^M\in\hat{\mathcal{H}}$, there exist a stopping time $\tau(\hat{r}_0^M)>0$ a.s., a point $z_0\in\mathbb{R}^n$, $d+1$ smooth vector fields $a,b_1\ldots,b_d$, defined on a neighborhood $\mathcal{Z}\subseteq\R^n$ of $z_0$, and a function $G:\mathcal{Z}\to\hat{\mathcal{H}}$ satisfying Assumption \ref{assumption - immersion condition}, such that  
\[
\hat{r}_t=G(Z_t)
\qquad\text{ for all }
t\in[0,\tau(\hat{r}_0^M))
\] 
almost surely, where $(Z_t)_{t\geq0}$ is an $n$-dimensional state process given as the strong solution to
\[
dZ_t=a(Z_t)dt+b(Z_t)\circ dW_t,
\qquad
Z_0=z_0.
\]
We say that model $\mathcal{M}$ has a {\em finite-dimensional realization} (FDR) if it has an $n$-dimensional realization, for some $n\in\mathbb{N}$.
\end{defin}

In the above definition, $\hat{r}^M_0\in\hat{\mathcal{H}}$ corresponds to the initially observed term structures of risk-free and risk-sensitive forward rates, together with the vector of log-spreads.
Definition \ref{fdr def} is intimately related to the concept of $\hat{r}$-invariance (see Definition \ref{r invariance def}). Indeed, it is apparent that the existence of an FDR is equivalent to the existence of an $\hat{r}$-invariant parameterized family $G$. Therefore, in view of Theorem \ref{invariance-thm}, given a multi-curve interest rate model $\mathcal{M}$, the existence of an FDR amounts to the existence of a submanifold $\mathcal{G}\subseteq\hat{\mathcal{H}}$ such that $\hat{\mu}(G(z)),\hat{\sigma}(G(z))\in T_{G(z)}\mathcal{G}$, for each $G(z)\in U$, where $U$ is a neighborhood of $\hat{r}^M_0$ and $\hat{r}^M_0\in \mathcal{G}$. In other words, we are looking for the tangential submanifold $\mathcal{G}$ of the distribution $F:=\mathrm{span}\{\hat{\mu},\hat{\sigma}_1,\ldots,\hat{\sigma}_d\}$.

By \cite[Theorem 2.1]{bjork2001existence}, a tangential submanifold for a smooth distribution $F$ exists if and only if $F$ is involutive, i.e., if and only if the Lie bracket between two vector fields in $F$ lies in $F$. We recall that the Lie bracket $[v_1,v_2]$ between two vector fields $v_1$ and $v_2$ on $\hat{\mathcal{H}}$ is defined as
\[
[v_1,v_2](\hat{r}) := \partial_{\hat{r}}v_1(\hat{r}) v_2(\hat{r}) -\partial_{\hat{r}}v_2(\hat{r})v_1(\hat{r}).
\]
\cite[Theorem 2.1]{bjork2001existence} is an infinite-dimensional version of the Frobenius theorem and gives the existence of a tangential submanifold when the distribution $F$ generated by $\hat{\mu}$ and $\hat{\sigma}$ is involutive. However, in general $F$ is not involutive. Therefore, we must consider the smallest involutive distribution that contains $F$. Such distribution is called the Lie algebra of $F$. Since \cite[Theorem 2.1]{bjork2001existence} is an abstract result, it can be immediately applied to our multi-curve framework, leading to the following statement (see  \cite[Theorem B.3.2]{lanaro2019geometry} for additional details).

\begin{theorem}	\label{thm:FDR}
A model $\mathcal{M}$ has an FDR if and only if there exists a finite-dimensional tangential submanifold for $\hat{\mu},\hat{\sigma}_1,\ldots,\hat{\sigma}_d$. In turn, this is equivalent to
\begin{equation}
\label{finite-dimensional realizations condition}
\mathrm{dim}[\mathcal{L}]
<+\infty,
\end{equation}
where $\mathcal{L}:=\{\hat{\mu},\hat{\sigma}_1,\ldots\hat{\sigma}_d\}_{\mathrm{LA}}$ denotes the Lie algebra generated by $\hat{\mu},\hat{\sigma}_1,\ldots,\hat{\sigma}_d$.
\end{theorem}

As long as condition \eqref{finite-dimensional realizations condition} is satisfied, an FDR can be constructed. To this end, we outline the general construction strategy provided in  \cite{bjork2002construction}, proceeding along the following three steps:
\begin{enumerate}
\item[I.] choose a finite number of independent vector fields $\xi_1,\ldots\xi_n$  which span $\{\hat{\mu},\hat{\sigma}_1,\ldots\hat{\sigma}_d\}_{\mathrm{LA}}$;
\item[II.] compute the invariant manifold 
\begin{equation}
\label{parameterized function G - finite-dimensional dimension}
G(z_1,\dots,z_n)=e^{\xi_nz_n}\cdots e^{\xi_1z_1}\hat{r}^M_0,
\end{equation}
where $e^{\xi_nz_n}$ denotes the integral curve of the vector field $\xi_n$ passing through $z_n$;
\item[III.] define the state process $(Z_t)_{t\geq0}$ taking values in $\R^n$ such that $\hat{r}=G(Z)$, making the following ansatz for the Stratonovich dynamics of the process $(Z_t)_{t\geq0}$:
\begin{equation}
\label{state-space process}
dZ_t=a(Z_t)dt+b(Z_t)\circ dW_t,
\end{equation}
where the vector fields $a$ and $b$ are determined by the following conditions:
\begin{equation}
\label{conditions on the vector field a and b}
\begin{aligned}
\partial_zG(z)a(z)&=\hat{\mu}(G(z)),\\ 
\partial_zG(z)b_i(z)&=\hat{\sigma}_i(G(z)),\qquad\text{for } i = 1,\ldots,d.
\end{aligned}
\end{equation}
The uniqueness of the vector fields $a$ and $b_i$, for $i=1,\ldots,d$, satisfying \eqref{conditions on the vector field a and b} is guaranteed since $G$ satisfies Assumption \ref{assumption - immersion condition}, hence it is a local diffeomorphism.
\end{enumerate}

The above methodology will be illustrated in Sections \ref{subsection - example constant volatility} and \ref{subsection - example constant direction volatility} in the case of constant volatility models and constant direction volatility models, respectively. In both cases, we will make use of the FX analogy discussed in Remark \ref{rem:FX}, which  enables us to adapt some arguments first used in \cite{slinko2010finite} in the context of an HJM model comprising a domestic and a foreign economy.

\subsection{Constant volatility models} 
\label{subsection - example constant volatility}

Let us first consider the case where the volatility $\hat{\sigma}(\hat{r})$ does not depend on $\hat{r}$, so that $\sigma^0,\sigma^1,\ldots,\sigma^m$ are constants in $\mathcal{H}^d$, while $\beta^1,\ldots,\beta^m$ are elements of $\mathbb{R}^d$ (where $d$ is the dimension of the driving Brownian motion). 
In this case, it holds that
\begin{equation}
\label{constant coefficients}
\hat{\mu}(\hat{r})=
\begin{pmatrix}
\textbf{F}r^0+\sigma^0\cdot\textbf{H}\sigma^0\\ 
\textbf{F}r^1+\sigma^1\cdot\textbf{H}\sigma^1-\beta^1\cdot\sigma^1\\ 
\vdots\\ 
\textbf{F}r^m+\sigma^m\cdot\textbf{H}\sigma^m-\beta^m\cdot\sigma^m\\ 
\textbf{B}r^0-\textbf{B}r^1-\frac{1}{2}||\beta^1||^2\\ 
\vdots\\ 
\textbf{B}r^0-\textbf{B}r^{m}-\frac{1}{2}||\beta^m||^2
\end{pmatrix}
\qquad\text{ and }\qquad
\hat{\sigma}(\hat{r})=
\begin{pmatrix}
\sigma^0\\ 
\sigma^1\\ 
\vdots\\ 
\sigma^m\\ 
\beta^1\\ 
\vdots\\ 
\beta^m
\end{pmatrix}.
\end{equation}

To determine the Lie algebra of $\mathrm{span}\{\hat{\mu},\hat{\sigma}_1,\ldots,\hat{\sigma}_d\}$, we compute the successive Lie brackets $[\hat{\mu},\hat{\sigma}](\hat{r})$ between $\hat{\mu}$ and $\hat{\sigma}$. This can be easily done by observing that $\partial_{\hat{r}}\hat{\sigma}(\hat{r}) = 0$, while 
\begin{equation}
\label{drift derivative constant}
\partial_{\hat{r}}\hat{\mu}(\hat{r})=
\begin{pmatrix}
\textbf{F}		&	0		&	0		&	\cdots	&	0		&	0		&	\cdots	&0\\ 
0			&	\textbf{F}	&	0		&	\cdots	&	0		&	0		&	\cdots	&0\\ 
0			&	0		&	\textbf{F}	&	\cdots	&	0		&	0		&	\cdots	&0\\ 
\vdots		&	\vdots	&	\vdots	&			&	\vdots	&	\vdots	& 			&\vdots\\ 
0			&	0		&	0		&	\cdots	&	\textbf{F}	&	0		&	\cdots	&0\\ 
\textbf{B}		&	-\textbf{B}	&	0		&	\cdots	&	0		&	0		&	\cdots	&0\\ 
\textbf{B}		&	0		&	-\textbf{B}	&	\cdots	&	0		&	0		&	\cdots	&0\\ 
\vdots		&	\vdots	&	\vdots	&			&	\vdots	&	\vdots	&			&\vdots\\ 
\textbf{B}		&	0		&	0		&	\cdots	&	-\textbf{B}	&	0		&	\cdots	&0
\end{pmatrix}.
\end{equation}
Therefore, the Lie bracket of $\hat{\mu}$ and $\hat{\sigma}_i$, for each $i=1,\ldots,d$, is given by
\begin{displaymath}
[\hat{\mu},\hat{\sigma}_i](\hat{r})=\partial_{\hat{r}}\hat{\mu}(\hat{r})\hat{\sigma}_i(\hat{r})-\overbrace{\partial_{\hat{r}}\hat{\sigma}_i(\hat{r})}^{=0}\hat{\mu}(\hat{r})=\begin{pmatrix}
\textbf{F}\sigma^0_i\\ 
\textbf{F}\sigma^1_i\\ 
\vdots\\ 
\textbf{F}\sigma^m_i\\ 
\textbf{B}\sigma^0_i-\textbf{B}\sigma^1_i\\ 
\vdots\\ 
\textbf{B}\sigma^0_i-\textbf{B}\sigma^m_i
\end{pmatrix},
\end{displaymath}
which is constant on $\hat{\mathcal{H}}$. 
Hence, the only vector field in $\mathcal{L}=\{\hat{\mu},\hat{\sigma}_1,\ldots,\hat{\sigma}_d\}_{\mathrm{LA}}$ that is not constant is $\hat{\mu}$. Therefore, to determine explicitly $\mathcal{L}$ it suffices to compute the Lie bracket between $\hat{\mu}$ and the successive Lie bracket between $\hat{\mu}$ and $\hat{\sigma}_i$. 
These arguments lead to 
\begin{equation}
\label{Lie algebra with constant vol}
\mathcal{L}
= \mathrm{span}\bigl\{\hat{\mu},\hat{\sigma}_1,\dots,\hat{\sigma}_d,\nu_i^k ; k\in\mathbb{N}^*, i=1,\dots,d\bigr\},
\end{equation}
where (see \cite[Section 3.2]{lanaro2019geometry} for full details)
\[
\nu_i^k := \bigl(	\textbf{F}^k\sigma_i^0,	\textbf{F}^k\sigma_i^1,	\ldots,	\textbf{F}^k\sigma_i^m,	\textbf{B}\textbf{F}^{k-1}\sigma_i^0-\textbf{B}\textbf{F}^{k-1}\sigma_i^1,			\ldots,	\textbf{B}\textbf{F}^{k-1}\sigma_i^0-\textbf{B}\textbf{F}^{k-1}\sigma_i^m
\bigr)^{\top}.
\]

To state a necessary and sufficient condition for the validity of \eqref{finite-dimensional realizations condition} in constant volatility models, we need to recall the concept of quasi-exponential function (see \cite[Definition 2.2]{bjork2004geometry}).

\begin{defin}
A function $f$ is said to be {\em quasi-exponential} (QE) if it is of the form
\begin{displaymath}
f(x)=\sum_ie^{\gamma_ix}+\sum_je^{\alpha_jx}\bigl(p_j(x)\cos(\omega_jx)+q_j(x)\sin(\omega_jx)\bigr),
\end{displaymath}
where $\gamma_i$, $\alpha_j$, $\omega_j$ are real numbers and $p_i$, $q_j$ real polynomials.
\end{defin}

We recall that QE functions can be characterised as follows (see, e.g.,   \cite[Lemma 2.1]{bjork2004geometry}).

\begin{lemma}
\label{QE functions characterization}
A function $f$ is $QE$ if and only if it is a component of the solution of a vector-valued linear ODE with constant coefficients, i.e., it holds that $f^{(n)}=\sum_{i=0}^{n-1}\gamma_i f^{(i)}$, where $f^{(i)}$ denotes the i-th order derivative of $f$.
\end{lemma}

\subsubsection{Existence of FDRs}
\label{subsubsection - existence of fdr constant}

In view of the FX analogy discussed in Remark \ref{rem:FX}, a straightforward adaptation of \cite[Proposition 3.2]{slinko2010finite}   yields the following characterization of the existence of FDRs in the constant volatility case (see also \cite[Theorem 3.2.3]{lanaro2019geometry} for a detailed proof).

\begin{proposition}
\label{characterization FDRs constant vol}
A multi-curve interest rate model $\mathcal{M}$ with constant volatility admits an FDR if and only if the function $\sigma_i^j$ is QE, for every $i=1,\ldots,d$ and $j=0,1,\ldots,m$. 
\end{proposition}

If $\sigma_i^j$ is QE, for all $i=1,\ldots,d$, $j=0,1,\ldots,m$,
the Lie algebra $\mathcal{L}=\{\hat{\mu},\hat{\sigma}_1,\ldots,\hat{\sigma}_d\}_{\mathrm{LA}}$ satisfies
\begin{equation}
\label{dimension of L}
\text{dim}[\mathcal{L}]\leq 1+\sum_{i=1}^d(1+n_i)=:n,
\end{equation}
where $n_i:= \mathrm{dim}[\mathrm{span}\{\nu_i^k; k\in\mathbb{N}^*\}]$, for each $i = 1,\dots,d$. 
As a consequence of Lemma \ref{QE functions characterization}, for each $i=1,\ldots,d$, there exists an annihilator polynomial 
\begin{equation}
\label{polynomial annihilator}
M_i(\textbf{F}):= \sum_{h = 0}^{n_i} \alpha^h_i \textbf{F}^h
\end{equation} 
such that $M_i(\textbf{F})\sigma^j_i = 0$ for all $j = 0,1,\ldots,m$. By \eqref{Lie algebra with constant vol}, the tangential manifold of dimension $n$ is obtained by the composition of the integral curves of $\hat{\mu}$, $\hat{\sigma}_i$, $\nu^k_i$, for  $k = 1,\ldots,n_i$ and $i = 1,\ldots,d$.

In the following, we will use this notation for a generic element $z$ of the state space $\R^{n}$:
\be\label{notation}
z:= (z^0,z^0_1,\ldots,z^{n_1}_1,z^0_2,\ldots,z^{n_2}_2,\ldots,z^0_d,\ldots,z^{n_d}_d)\in\R^{n}.
\ee

\subsubsection{Construction of FDRs}

In order to construct explicitly an FDR, we apply the three-step procedure outlined at the end of Section \ref{subsection - the strategy}. In view of Theorem \ref{characterization FDRs constant vol}, it suffices to compute the integral curve of every vector field generating the Lie algebra given in \eqref{Lie algebra with constant vol}. Then, in step II, we compose the integral curves, thus obtaining the tangential manifold. Finally, in step III, we invert the consistency condition and obtain the following result, which generalizes to the multi-curve setting \cite[Proposition 5.2]{bjork2001existence} and \cite[Proposition 3.5]{slinko2010finite} (see \cite[Section 3.2.1]{lanaro2019geometry} for a detailed proof). 
We make use of the notation $S^j(x) := (\int_0^x \sigma^j_i(s) ds)_{i = 1,\dots,d}$, for $j = 0,1,\ldots,m$.

\begin{proposition}
Suppose that a multi-curve interest rate model $\mathcal{M}$ with constant volatility admits an FDR. In this case, the invariant manifold generated by $\hat{r}^M_0$ is parameterized as follows:
\begin{align*}
G^j(z,x)		&=	r^M_j(x+z^0)+\sum_{i=1}^d\sum_{k=0}^{n_i}\mathbf{F}^k\sigma^j_i(x)z^k_i+\frac{1}{2}\bigl(||S^j(x+z^0)||^2-||S^j(x)||^2\bigr)\\ 
				&\quad-(1-\delta^j_0)\sum_{i = 1}^d\beta^j_i\int_x^{x+z^0} \sigma^j_i(s)ds,\quad j = 0,1,\ldots,m;\\ 
G^{m+j}(z)	&=	y^M_{j} + \int_0^{z^0}\bigl(r^M_0(s)-r^M_j(s)\bigr)ds + \sum_{i=1}^d\sum_{k=1}^{n_i}(\mathbf{B}\mathbf{F}^{k-1}\sigma^0_i-\mathbf{B}\mathbf{F}^{k-1}\sigma^{j}_i)z^k_i+\sum_{i=1}^d\beta^j_iz^0_i\\ 
				&\quad +\frac{1}{2}\int_0^{z^0}\bigl(||S^0(s)||^2-||S^{j}(s)||^2\bigr)ds+\sum_{i = 1}^d\beta^j_i\int_0^{z^0}S^{j}_i(s)ds-\frac{1}{2}||\beta^{j}||^2z^0,\quad j = 1,\dots,m,
\end{align*}
where $\delta^j_0$ denotes the Kronecker delta between indexes $0$ and $j$.
In addition, the coefficients of the dynamics of the state process $(Z_t)_{t\geq0}$ in \eqref{state-space process} are given by
\begin{equation}
\label{constant vol - Z coefficients}
\begin{cases}
a^0=1,\\ 
a^0_i=0,							&\quad i=1,\dots,d,\\ 
a^k_i=z^{k-1}_i+z^{n_i}_i\alpha^k_i,		&\quad k=1,\dots,n_i,\  i=1,\dots,d,\\ 
b^0_i=0,\\ 
b_{i,i}^0=1,\\ 
b^k_{h,i}=0,						&\quad  h=1,\dots,d,\ h\neq i,\ k=1,\dots,n_i.
\end{cases}
\end{equation}
according to the notation introduced in \eqref{notation}. In \eqref{constant vol - Z coefficients}, the term $\alpha^k_i$ denotes the $k$-th coefficient of the annihilator polynomial $M_i$ given in \eqref{polynomial annihilator}, for $k=1,\dots,n_i$ and $i = 1,\dots,d$. 
\end{proposition}

\subsection{Constant direction volatility models}
\label{subsection - example constant direction volatility}

In this section, we consider a more general class of multi-curve interest rate models with non-constant volatility, exploiting the analogy with the two-economy HJM setup analysed in \cite{slinko2010finite} (see Remark \ref{rem:FX}). We aim at providing conditions ensuring the existence of FDRs for a model $\mathcal{M}$ determined by a volatility of the following form:
\begin{equation}
\label{constant direction volatility term}
\hat{\sigma}_i(\hat{r})(x)=
\bigl(
\varphi^0_i(\hat{r})\lambda_i^0(x),
\varphi^1_i(\hat{r})\lambda_i^1(x),
\ldots,
\varphi^m_i(\hat{r})\lambda^m_i(x),
\beta^1_i(\hat{r}),
\cdots,
\beta^m_i(\hat{r})
\bigr),
\end{equation}
for $i=1,\ldots,d$, where $\lambda^j_i\in\mathcal{H}$ and $\varphi^j_i$ and $\beta^j_i$ are smooth (in the Fr\'echet sense) scalar vector fields on $\hat{\mathcal{H}}$, for every $i=1,\ldots,d$ and $j=0,1,\ldots,m$. 
We introduce the following assumption.

\begin{assumption}
\label{non zero condition}
For all $i=1,\ldots,d$ and $j=0,1,\ldots,m$, it holds that $\varphi^j_i(\hat{r})\neq0$ and $\beta^j_i(\hat{r})\neq0$, for every $\hat{r}\in\hat{\mathcal{H}}$.
\end{assumption}


As a first step, we need to compute explicitly the drift term in the Stratonovich dynamics \eqref{forward-rate system} of the joint process $\hat{r}$. To this effect, we notice that, for each $j=0,1,\ldots,m$,
\begin{displaymath}
\begin{split}
\partial_{\hat{r}}\hat{\sigma}(\hat{r})\hat{\sigma}(\hat{r})	&=\sum_{i=1}^d\Biggl{(}\sum_{h=0}^m\lambda^j_i\partial_{r^h}\varphi^j_i(\hat{r})\varphi^h_i(\hat{r})\lambda^h_i+\sum_{h=1}^m\lambda^j_i\partial_{Y^h}\varphi^j_i(\hat{r})\beta^h_i(\hat{r})\Biggr{)},\\ 
\sigma^j(\hat{r})\cdot\textbf{H}\sigma^j(\hat{r})			&
=\sum_{i=1}^d(\varphi^j_i(\hat{r}))^2\lambda^j_i\int_0^{\cdot}\lambda^j_i(s)ds.
\end{split}
\end{displaymath}
We denote by $\partial_{{r}^h}\varphi^j_i(\hat{r})[\lambda^h_i]$ the Fr\'echet derivative of $\varphi^j_i$ with respect to the variable $r^h$ computed at $\hat{r}$, acting on the vector $\lambda^h_i$, for each $h,j=0,1,\ldots,m$ and $i=1,\ldots,d$. We also define
\begin{equation}
\label{Dji}
D^j_i(x):=\lambda^j_i(x)\int_0^x\lambda^j_i(s)ds,
\qquad \text{ for }j=0,1,\ldots,m\text{ and }i=1,\ldots,d.
\end{equation}
Hence, the Stratonovich dynamics of the $2m+1$ components of \eqref{forward-rate system} are given by
\begin{displaymath}
\begin{split}
dr^j_t	&=\Bigg(\textbf{F}r^j_t+\sum_{i=1}^d(\varphi^j_i(\hat{r}_t))^2D^j_i-\frac{1}{2}\sum_{i=1}^d\lambda^j_i\Biggl{(}\sum_{h=0}^m\varphi^h_i(\hat{r}_t)\partial_{r^h}\varphi^j_i(\hat{r}_t)[\lambda^h_i]+\sum_{h=1}^m\partial_{Y^h} \varphi^j_i(\hat{r}_t)\beta^h_i(\hat{r}_t)\\ 
		&\qquad+2(1-\delta^0_j)\varphi^j_i(\hat{r}_t)\beta^j_i(\hat{r}_t)\Biggr{)}\Bigg)dt+\sum_{i=1}^d\varphi^j_i(\hat{r}_t)\lambda^j_i\circ dW_t,
\qquad \text{for }j = 0,1,\ldots,m,\\ 
dY^j_t	&=\Bigg(\textbf{B}r^0_t-\textbf{B}r^j_t-\frac{1}{2}\sum_{i=1}^d(\beta^j_i(\hat{r}_t))^2-\frac{1}{2}\sum_{i=1}^d\Biggl(\sum_{h=0}^m\partial_{r^h} \beta^j_i(\hat{r}_t)[\lambda^h_i]\varphi^h_i(\hat{r}_t)+\sum_{h=1}^m\partial_{Y^h}\beta^j_i(\hat{r}_t)\beta^h_i(\hat{r}_t)\Biggr)\Bigg)dt\\ 
		&\qquad+\sum_{i=1}^d\beta^j_h(\hat{r}_t)\circ dW_t,
\qquad\qquad\qquad\qquad\qquad\qquad\qquad\qquad
\text{for } j = 1,\dots,m.
\end{split}
\end{displaymath}
Therefore, the Stratonovich drift term $\hat{\mu}(\hat{r})$ can be explicitly written as follows:
\begin{small}
\[
\hat{\mu}(\hat{r})=
\begin{pmatrix}
\textbf{F}r^0+\sum_{i=1}^d(\varphi^0_i(\hat{r}))^2D^0_i-\frac{1}{2}\sum_{i=1}^d\lambda^0_i\bigl(\sum_{h=0}^m\varphi^h_i(\hat{r})\partial_{r^h}\varphi^0_i(\hat{r})[\lambda^h_i] +\sum_{h=1}^m\partial_{Y^h} \varphi^0_i(\hat{r})\beta^h_i(\hat{r})\bigr)\\ 
\textbf{F}r^1+\sum_{i=1}^d(\varphi^1_i(\hat{r}))^2D^1_i-\frac{1}{2}\sum_{i=1}^d\lambda^1_i\bigl(\sum_{h=0}^m\varphi^h_i(\hat{r})\partial_{r^h}\varphi^1_i(\hat{r})[\lambda^h_i]\\ 
\qquad\qquad\qquad\qquad\qquad\qquad\qquad
+\sum_{j=1}^m\partial_{Y^h} \varphi^1_i(\hat{r})\beta^h_i(\hat{r})+2\varphi^1_i(\hat{r})\beta^1_i(\hat{r})\bigr)\\ 
\vdots\\ 
\textbf{F}r^m+\sum_{i=1}^d(\varphi^m_i(\hat{r}))^2D^m_i-\frac{1}{2}\sum_{i=1}^d\lambda^m_i\bigl(\sum_{h=0}^m\varphi^h_i(\hat{r})\partial_{r^h}\varphi^m_i(\hat{r})[\lambda^h_i]\\ 
\qquad\qquad\qquad\qquad\qquad\qquad\qquad
+\sum_{h=1}^m\partial_{Y^h} \varphi^m_i(\hat{r})\beta^h_i(\hat{r})+2\varphi^m_i(\hat{r})\beta^m_i(\hat{r})\bigr)\\ 
\textbf{B}r^0-\textbf{B}r^1-\frac{1}{2}\sum_{i=1}^d(\beta^1_i(\hat{r}))^2-\frac{1}{2}\sum_{i=1}^d\bigl(\sum_{h=0}^m\partial_{r^h} \beta^1_i(\hat{r})[\lambda^h_i]\varphi^h_i(\hat{r})+\sum_{h=1}^m\partial_{Y^h}\beta^1_i(\hat{r})\beta^h_i(\hat{r})\bigr)\\ 
\vdots\\ 
\textbf{B}r^0-\textbf{B}r^m-\frac{1}{2}\sum_{i=1}^d(\beta^m_i(\hat{r}))^2-\frac{1}{2}\sum_{i=1}^d\bigl(\sum_{h=0}^m\partial_{r^h} \beta^m_i(\hat{r})[\lambda^h_i]\varphi^h_i(\hat{r})+\sum_{h=1}^m\partial_{Y^h}\beta^m_i(\hat{r})\beta^h_i(\hat{r})\bigr)
\end{pmatrix}.
\]
\end{small}

The complexity of the Stratonovich drift $\hat{\mu}(\hat{r})$ arises from the fact that, for a multi-curve interest rate model $\mathcal{M}$ with volatility as in \eqref{constant direction volatility term}, each component of the volatility function $\hat{\sigma}(\hat{r})$ can depend on the whole vector of forward rate and log-spread processes.

\subsubsection{Existence of FDRs}

Differently from Section \ref{subsection - example constant volatility}, in the case of constant direction volatility the computation of the integral curve of $\hat{\mu}$ turns out to be more difficult. To overcome this obstacle, we will derive conditions ensuring that a distribution larger than $\{\hat{\mu},\hat{\sigma}_1,\dots,\hat{\sigma}_d\}_{\mathrm{LA}}$ is finite-dimensional. This will provide a sufficient condition for the existence of FDRs.

As a preliminary, we denote by $E_j$ the $j$-th element of the canonical basis of $\hat{\mathcal{H}}$, for each $j = 0,1,\ldots,2m$, and  introduce the following family of vector fields:
\begin{equation}
\label{N constant direction vol}
\mathcal{N}:=\bigl\{\xi^0,\xi^h_i,\eta^h_i,\gamma_k ; h=0,1,\ldots,m,\,i=1,\ldots,d, \,k=1,\ldots,m\bigr\},
\end{equation}
where
\begin{displaymath}
\xi^0:=
\begin{pmatrix}
\textbf{F}r^0\\ 
\textbf{F}r^1\\ 
\vdots\\ 
\textbf{F}r^m\\ 
\textbf{B}r^0-\textbf{B}r^1\\ 
\vdots\\ 
\textbf{B}r^0-\textbf{B}r^m
\end{pmatrix},
 \quad \xi_i^h:=\lambda^h_i E_h,\quad 
 \eta^h_i:=D^h_iE_h,\quad \gamma_k:=E_{m+k}.
\end{displaymath}
Making use of this notation, we can write
\begin{align}
\label{mu is inside}
\hat{\mu}(\hat{r})		&=\xi^0+\sum_{h=0}^m\sum_{i=1}^d\bigl{(}(\varphi^h_i(\hat{r}))^2\eta^h_i-\kappa^h_i\xi^h_i\bigr{)}-\sum_{h=1}^m\zeta^h\gamma_h,\\ 
\label{volatility is inside}
\hat{\sigma}_i(\hat{r})	&=\sum_{h=0}^m\varphi^h_i(\hat{r})\xi^h_i+\sum_{h=1}^m\beta^h_i(\hat{r})\gamma_h,
\end{align}
for $i =1,\ldots,d$, where 
\begin{align*}
\kappa^j_i	&:=\frac{1}{2}\Bigg(\sum_{h=0}^m\varphi^h_i(\hat{r})\partial_{r^h} \varphi^j_i(\hat{r})[\lambda^h_i]+\sum_{h=1}^m\bigl(\partial_{Y^h} \varphi^j_i(\hat{r})\beta^h_i(\hat{r})\bigr)+2(1-\delta^0_j)\varphi^j_i(\hat{r})\beta^j_i(\hat{r})\Bigg),\\ 
\zeta^j	&:=\frac{1}{2}\Bigg(\sum_{i=1}^d(\beta^j_i(\hat{r}))^2+\sum_{i=1}^d\Bigg{(}\sum_{h=0}^m\partial_{r^h} \beta^j_i(\hat{r})[\lambda^h_i]\varphi^h_i(\hat{r})+\sum_{h=1}^m\partial_{Y^h} \beta^j_i(\hat{r})\beta^h_i(\hat{r})\Biggr{)}\Bigg).
\end{align*}
It is easily seen that \eqref{mu is inside}-\eqref{volatility is inside} directly imply that $\mathcal{L}=\{\hat{\mu},\hat{\sigma}_1,\ldots,\hat{\sigma}_d\}_{\mathrm{LA}}\subseteq \mathcal{L}^1$, where
\begin{equation}
\label{L1 cdv}
\mathcal{L}^1 := \mathcal{N}_{\mathrm{LA}}.
\end{equation}
Consequently, if $\mathcal{L}^1$ is finite-dimensional, then also $\mathcal{L}$ is finite-dimensional. Moreover, the Lie algebra $\mathcal{L}^1$ has a much simpler structure than $\mathcal{L}$.
Similarly to the case of constant volatility models analysed in Section \ref{subsection - example constant volatility}, it can be proved that $\mathcal{L}^1$ is finite-dimensional if and only if $\lambda^j_i$ is a QE function, for every $j=0,1,\ldots,m$ and $i=1,\ldots,d$ (compare with Proposition \ref{characterization FDRs constant vol} and see also \cite[Section 3.2.1]{lanaro2019geometry} for additional details).
We have thus obtained the following result, which in particular provides an extension of \cite[Proposition 4.2]{slinko2010finite} to the multi-curve setting.

\begin{proposition}
\label{sufficient condition constant direction theorem}
If $\lambda^j_i$ is a QE function, for every $j=0,1,\ldots,m$ and $i=1,\ldots,d$, then the Lie algebra $\mathcal{L}=\{\hat{\mu},\hat{\sigma}_1,\ldots,\hat{\sigma}_d\}_{\mathrm{LA}}$ is finite-dimensional.
\end{proposition}

\subsubsection{Construction of FDRs}

In view of Theorem \ref{thm:FDR}, the result of Proposition \ref{sufficient condition constant direction theorem} provides a sufficient condition for the existence of an FDR. We now turn to the construction of FDRs. In this section, to guarantee the existence of FDRs, we shall work under the following assumption.

\begin{assumption}
\label{assumption cdvol}
For every $j=0,1,\dots,m$ and $i=1,\dots,d$, the function $\lambda^j_i$ is QE.
\end{assumption}

Under Assumption \ref{assumption cdvol}, it can be easily verified that the functions $D^j_i$ given in \eqref{Dji} are QE. By Lemma \ref{QE functions characterization}, there exist $n^j_i\in\mathbb{N}$ and $p^j_i\in\mathbb{N}$, for all $i=1,\dots,d$ and $j=0,1,\dots,m$, such that 
\[
\textbf{F}^{n^j_i}\lambda^j_i=\sum_{k=0}^{n^j_i-1}c^j_{k,i}\textbf{F}^k\lambda^j_i
\qquad\text{ and }\qquad
\textbf{F}^{p^j_i}D^j_i =\sum_{k=0}^{p^j_i-1}d^j_{k,i}\textbf{F}^kD^j_i,
\] 
for suitable constants $c^j_{k,i}$ and $d^{j}_{k,i}$. In this case, due to the definition of the Lie algebra $\mathcal{L}^1$, the dimension of $\mathcal{L}^1$ is bounded by $n:= m+1+\sum_{i=1}^d\sum_{j=0}^m(n^j_i+p^j_i)$. 

We introduce the following notation for a generic vector $z\in\mathbb{R}^{n}$, obtained by concatenating all elements of the vectors $(x^j)_{j}\in \R^{m+1}$, $(z^j_{k,i})_{i,j,k}\in\mathbb{R}^{\sum_{i=1}^d\sum_{j=0}^mn^j_i}$ and $(x^j_{k,i})_{i,j,k}\in\mathbb{R}^{\sum_{i=1}^d\sum_{j=0}^mp^j_i}$:
\begin{equation}
\label{notation for coordinates 2} 
z= (x^0,\ldots,x^m,z^0_{0,1},\ldots,z^0_{n^0_1-1,1}, \ldots,z^m_{0,d},\ldots,z^m_{n^m_d-1,d},x^0_{0,1},\ldots,x^0_{p^0_1-1,1},\ldots,x^m_{0,d},\ldots,x^m_{p^m_d-1,d}).
\end{equation}

As explained in Section \ref{subsubsection - existence of fdr constant},  the tangential manifold of $\mathcal{L}^1$ can be constructed as follows:
\begin{equation}
\label{constant direction volatility tangential function}
G(z):=\prod_{i,j,k,h,l} e^{\textbf{F}^k\lambda^j_iE_jz^j_{k,i}}e^{\textbf{F}^hD^j_iE_jx^j_{h,i}}e^{\gamma_lx^l}e^{\xi^0x^0}\hat{r}^M_0,
\end{equation}
for an arbitrary initial point $\hat{r}^M_0\in\hat{\mathcal{H}}$. 
In the following, we denote by $r^M_j$ the $j$-th component of $\hat{r}^M_0$, for $j=0,1,\ldots,m$. 
The components of the function $G$ introduced in \eqref{constant direction volatility tangential function} are given by
\begin{equation}
\begin{aligned}
\label{finite-dimensional realizations mapping 2}
G^j(z,x)	&=r^M_j(x^0+x)+\sum_{i=1}^d\Bigg(\sum_{k=0}^{n^j_i-1}z^j_{k,i}\textbf{F}^k\lambda^j_i(x)+\sum_{k=0}^{p^j_i-1}x^j_{k,i}\textbf{F}^k\bigl{(}\lambda^j_i(x)D^j_i(x)\bigr{)}\Bigg),
\qquad j=0,\dots,m,\\ 
G^j(z)	&=y^M_j+\int_0^{x^0}(r^M_{0,0}(s)-r^M_{0,j}(s))ds+x^j,
\quad\qquad\qquad\qquad\qquad\qquad\qquad j=m+1,\dots,2m.
\end{aligned}
\end{equation}

At this point, we can determine the coefficients $a$ and $b$ of the state process $(Z_t)_{t\geq0}$ by requiring that condition \eqref{conditions on the vector field a and b} holds.
For simplicity of notation, we shall omit to denote the dependence on $z$ in the functions $a$ and $b$ and adopt a notation consistent with  \eqref{notation for coordinates 2}:
\[
a	= (a^0,\ldots,a^m,a^0_{0,1},\ldots,a^m_{n^m_d-1,d},\widetilde{a}^0_{0,1},\ldots,\widetilde{a}^m_{p^m_d-1,d})
\qquad\text{ and }\qquad
b	=(b_1,\ldots,b_d)^\top,
\]
where $b_h$, for each $h = 1,\ldots,d$, has the same representation of the vector $a$.
To determine $a$ and $b$, we must invert the consistency condition between the coefficients of the model and tangential manifold, as described in step III of the procedure outlined in Section \ref{subsection - the strategy}. In this way we obtain
\[
\begin{cases}
a^0				&=	1,\\ 
a^j_{0,i}			&=z^j_{n^j_i-1,i}c^j_{0,i}-\frac{1}{2}\Bigl{(}\sum_{h=0}^m\varphi^h_i(G(z))\partial_{r^h} \varphi^j_i(G(z))[\lambda^j_i]\\ 
				&\quad+\sum_{h=1}^m\beta^h_i(G(z))\partial_{Y^h} \varphi^j_i(G(z))+(1- \delta^j_0)2\varphi^j_i(G(z))\beta^j_i(G(z))\Bigr{)},\qquad\qquad j = 0,\dots,m,\\ 
a^j_{k,i}			&=z^j_{k-1,i}+z^j_{n^j_i-1,i}c^j_{k,i},\qquad\qquad\qquad\qquad\qquad\qquad\qquad\quad 
k=1,\dots, n^j_i-1, j = 0,\dots,m,\\ 
\widetilde{a}^j_{0,i}	&=(\varphi^j_i(G(z)))^2+x^j_{p^j_i-1,i}d^j_{0i},\qquad\qquad\qquad\qquad\qquad\qquad\qquad\qquad\qquad\qquad\, j = 0,\dots,m,\\ 
\widetilde{a}^j_{k,i}	&=x^j_{k-1,i}+x^j_{p^j_i-1,i}d^j_{k,i},\qquad\qquad\qquad\qquad\qquad\qquad\qquad\;\, k=1,\dots,p^j_i-1,\ j = 0,\dots,m,\\ 
a^{j}				&=\sum_{i=1}^d\Bigl(\sum_{k=0}^{n^0_i-1}z^0_{k,i}\textbf{F}^k\lambda^0_i(0)-\sum_{k=0}^{n^{j}_i-1}z^j_{k,i}\textbf{F}^k\lambda^j_i(0)-\frac{1}{2}(\beta^{j}_i(G(z)))^2\\ 
				& \quad -\frac{1}{2}\Bigl{(}\sum_{h=0}^m\partial_{r^h} \beta^{j}_iG(z))[\lambda^h_i]\varphi^h_i(G(z))+\sum_{h=1}^m\partial_{Y^h} \beta^{j}_i(G(z))\beta^h_i(G(z))\Bigr{)}\Bigr),\quad j = 1,\dots,m.
\end{cases}
\]
We follow the same procedure in order to determine $b$. We obtain, for each $h = 1,\dots,d$,
\[
\begin{cases}
b^0_h			&=0,\\ 
b^j_{0,h,h}		&=\varphi^j_h(G(z)),\\ 
b^j_{0,i,h}			&=0,\quad i\neq h,\quad i=1,\dots,d,\\ 
b^j_{k,i,h}			&=0,\quad k=1,\dots,n^j_i-1,\\ 
\widetilde{b}^j_{k,i,h}	&=0,\quad k=0,\dots,p^j_i-1,\\ 
b^0				&=0,\\ 
b^j_{h}			&=\beta^j_h(G(z)),\qquad j=1,\dots,m.\\ 
\end{cases}
\]

In conclusion, under Assumption \ref{assumption cdvol}, an FDR for a model $\mathcal{M}$ defined by a volatility $\hat{\sigma}$ of the form \eqref{constant direction volatility term} can be determined by the immersion $G$ defined in \eqref{finite-dimensional realizations mapping 2} and by the finite-dimensional state process $(Z_t)_{t\geq0}$,  whose drift and volatility coefficients have been explicitly computed above.
 
\begin{remark}
Assumption \ref{assumption cdvol} only represents a sufficient condition for the existence of FDRs. In the more specific case where $\varphi^j_i = \varphi_i$, for all $i = 1,\dots,d$ and $j =0,1,\dots,m$, necessary and sufficient conditions can be obtained for the existence of FDRs. We refer the interested reader to \cite[Section 3.3.2]{lanaro2019geometry} for a detailed analysis of this situation.
\end{remark}

\subsubsection{An example: Hull White model with non-constant volatility of the log-spread processes} 
\label{sec:example_HW}

To exemplify the construction of FDRs for constant direction volatility models described above, we consider a simple model with $m=2$ (i.e., we consider two distinct tenors for risk-sensitive rates, together with the risk-free rate) and $d=3$, with the following volatility structure:
\[
\hat{\sigma}(\hat{r})=
\begin{pmatrix}
\sigma^0e^{-a^0x} 	& 0                                	& 0                                \\ 
0                               	& \sigma^1e^{-a^1x}		& 0                                \\ 
0                           	& 0                            		& \sigma^2e^{-a^2x} \\
\beta^1_1                	& \beta^1_2Y^1_t       	& 0                                \\ 
\beta^2_1                	& 0                                	& \beta^2_3Y^2_t
\end{pmatrix},
\]
where $\beta^j_i,\sigma^j,a^j$ are positive constants, for $j=0,1,2$ and $i=1,2,3$. This volatility specification is consistent with the empirically observed fact that large values of the spreads are typically associated with high volatility, as large values of the spreads tend to occur in periods of financial turmoil. Moreover, we allow for correlation between forward rates and log-spread processes. 

We aim at constructing explicitly an FDR for this three-curve model.
We assume that the initial point $\hat{r}^M_0=(r^M_0,r^M_1,r^M_2,y^M_1,y^M_2)\in\hat{\mathcal{H}}$ is such that $y^M_1,y^M_2\neq0$. By continuity of $Y^j$, we can therefore assume that $Y^j_t\neq0$ for all $t\in[0,\tau)$, for some stopping time $\tau>0$ a.s. In this case, Assumption \ref{non zero condition} holds true and the Lie algebra generated by the coefficients of the model is finite-dimensional. 
Since the volatility functions of the forward rate processes are QE functions, Proposition \ref{sufficient condition constant direction theorem} can be applied. We therefore obtain that the Lie algebra
\begin{displaymath}
\mathcal{L}^1=\mathrm{span}\bigl{\{}\xi^0,\ \textbf{F}^k\lambda^j_iE_j,\ \textbf{F}^kD^j_iE_j,\ \gamma_l\ :\ j=0,1,2,\  i=1,2,3,\ l=1,2,\ k\in\mathbb{N}\bigr{\}}
\end{displaymath}
is finite-dimensional and contains $\{\hat{\mu},\hat{\sigma}_1,\hat{\sigma}_2,\hat{\sigma}_3\}_{\mathrm{LA}}$. In view of Proposition \ref{sufficient condition constant direction theorem}, this ensures the existence of FDRs for the model considered in this example.

To construct the FDRs, we first compute the dimension of $\mathrm{span}\{\textbf{F}^k\lambda^j_iE_j,\ \textbf{F}^kD^j_iE_j\ :\ k\in\mathbb{N}\}$. For each $i=1,2,3$, the polynomial annihilator $M^j_i$ of $\sigma^j_i$ is trivially $M^j_i(\gamma)=\gamma$ for $j\neq i-1$, while it is $M^j_i(\gamma)=\gamma+a^j$ for $j=i-1$. 
Analogously, since 
\[
D^j_i(x) =\frac{(\sigma^j)^2}{a^j}(e^{-a^jx}-e^{-2a^jx}),
\qquad\text{ for }j=i-1,
\]
and $0$ otherwise, the polynomial annihilator of $D^j_i$ is $P^j_i(\gamma) = \gamma$ for $j\neq i-1$, while for $j=i-1$ it is given by $P^{i-1}_i(\gamma) = \gamma^2 + 3a^j\gamma+2(a^j)^2$, for $i=1,2,3$. Therefore, we have that
\begin{align*}
n^j_i	&= \text{deg}(M^j_i) = 1,\quad \text{for }j = 0,1,2,\ i = 1,2,3,\\ 
p^j_i &= \text{deg}(P^j_i) = 
\begin{cases} 
2,\quad \text{for }j = i-1,\ i = 1,2,3,\\
1,\quad \text{otherwise}.
\end{cases}
\end{align*}
Referring to the family of vector fields introduced in \eqref{N constant direction vol}, we have that $\xi^j_i=0$ and $\eta^j_i=0$, for all $i=1,2,3$ and $j\neq i-1$. Hence, these vectors do not contribute to generate the Lie algebra $\mathcal{L}^1$, which has dimension $n= 2+1+\sum_{i=1}^3(n^{i-1}_i+p^{i-1}_i) = 12$.
In analogy with notation \eqref{notation for coordinates 2}, we denote a generic vector $z\in\mathbb{R}^{n}$ by $z = (x^0,x^1,x^2, z^0, z^1, z^2, x^0_0, x^0_1, x^1_0, x^1_1, x^2_0, x^2_1)$. 

The tangential manifold is given by the image of the function introduced in equation \eqref{constant direction volatility tangential function}. 
In the present example, the latter is given by
\begin{displaymath}
G(z)=\prod_{i,h,l} e^{\sigma^{i-1}_iE_{i-1}z^{i-1}}e^{\textbf{F}^hD^{i-1}_iE_{i-1}x^{i-1}_{h}}e^{\gamma_lx^l}e^{\xi^0x^0}\hat{r}^M_0.
\end{displaymath}

The process $Z_t = (X^0_t,X^1_t,X^2_t,Z^0_{t},Z^1_{t},Z^2_{t},X^0_{0,t},X^0_{1,t},X^1_{0,t},X^1_{1,t},X^2_{0,t},X^2_{1,t})$ satisfies \eqref{state-space process}, where:
\[
\begin{small}
a(Z_t) = 
\begin{pmatrix}
1\\ 
Z^0_{t}\sigma^0+(\sigma^0)^2X^0_{1,t}-Z^1_{t}\sigma^1-(\sigma^1)^2X^1_{1,t}-\frac{1}{2}(\beta^1_2)^2\bigl(y^M_1+\int_0^{X^0_t}(r^M_0(s)-r^M_1(s))ds\\ 
\qquad\qquad+X^1_t\bigr)\bigl(y^M_1+\int_0^{X^0_t}(r^M_0(s)-r^M_1(s))ds+X^1_t+1\bigr)-\frac{1}{2}(\beta^1_1)^2\\ 
Z^0_{t}\sigma^0+(\sigma^0)^2X^0_{1,t}-Z^2_t\sigma^2-(\sigma^2)^2X^2_{1,t}-\frac{1}{2}(\beta^2_3)^2\bigl(y^M_2+\int_0^{X^0_t}(r^M_0(s)-r^M_2(s))ds\\ 
\qquad\qquad+X^2_t\bigr)\bigl(y^M_2+\int_0^{X^0_t}(r^M_0(s)-r^M_2(s))ds+X^2_t+1\bigr)-\frac{1}{2}(\beta^2_2)^2\\ 
-a^0Z^0_{t}\\ 
-Z^1_{t}a^1-\beta^1_2\bigl(y^M_1+\int_0^{X^0_t}(r^M_0(s)-r^M_1(s))ds+X^1_t\bigr)\\ 
-Z^2_{t}a^2 -\beta^2_3\bigl(y^M_2+\int_0^{X^0_t}(r^M_0(s)-r^M_2(s))ds+X^2_t\bigr)\\ 
-2(a^0)^2X^0_{1,t}+1\\ 
X^0_{0,t}-3a^0X^0_{1,t}\\ 
-2(a^1)^2X^1_{1,t}+1\\ 
X^1_{0,t}-3a^1X^1_{1,t}\\ 
-2(a^2)^2X^2_{1,t}+1\\ 
X^2_{0,t}-3a^2X^2_{1,t}
\end{pmatrix},
\end{small}
\]
\[
\begin{small}
b(Z_t) =
\begin{pmatrix}
0&0&0\\ 
\beta^1_1&\beta^1_2\bigl(y^M_1+\int_0^{X^0_t}(r^M_0(s)-r^M_1(s))ds+X^1_t\bigr)&0\\ 
\beta^2_1&0&\beta^2_3\bigl(y^M_2+\int_0^{X^0_t}(r^M_0(s)-r^M_2(s))ds+X^2_t\bigr)\\ 
1&0&0\\ 
0&1&0\\ 
0&0&1\\ 
0&0&0\\ 
0&0&0\\ 
0&0&0\\ 
0&0&0\\ 
0&0&0\\ 
0&0&0\\ 
\end{pmatrix}.
\end{small}
\]

\subsection{Realizations through a set of benchmark forward rates and log-spreads}
\label{sec:benchmark}

Let us consider a model $\mathcal{M}$ admitting an $n$-dimensional realization. In general, the state process $(Z_t)_{t\geq0}$ of the FDR has no economic interpretation. In this section, we show that it is possible to construct another FDR of the same dimension $n$ determined by forward rates associated to a fixed set of benchmark maturities together with the log-spreads. This result is interesting in view of applications, where it is useful to find realizations for which the state process admits an economic interpretation and, in particular, can be readily deduced from market observables.

Given the existence of an $n$-dimensional realization, considering a point $\hat{r}=(r,y)\in\mathcal{G}\subseteq\hH$, we aim at constructing an FDR through the linear functional
\begin{equation}
\label{coordinates}
Z^h(r,y):=\sum_{j=0}^ma^h_j\,r^j(x_h)+\sum_{j=1}^mb^h_j\,y^j, 
\qquad \text{ for }h = 1,\ldots, n,
\end{equation}
where $(x_1,\ldots,x_n)$ is a fixed set of maturities and $a^h_j$ and $b^h_j$ are constants. 
The problem is equivalent to prove that $Z:\hH\rightarrow \R^n$ is a local system of coordinates for $\mathcal{G}$, for a suitable choice of the constants $a^h_j$ and $b^h_j$. 
The following proposition provides a positive answer to this problem, generalizing \cite[Theorem 3.3]{bjork2001existence} to the multi-curve setup. We remark that, differently from the classical single-curve case, in the multi-curve setup each component of the FDR determined by \eqref{coordinates} is given by a linear combination of $m+1$ forward rates and $m$ log-spreads.

\begin{proposition}
Let $\mathcal{M}$ be a multi-curve interest rate model admitting an $n$-dimensional realization. Then, for any vector $(x_1,\dots,x_n)\in\R^n$, where each $x_i$ is arbitrarily chosen except for a discrete subset of $\R_+$, the realization can be described by the inverse of the function \eqref{coordinates}, for suitable constants $a^h_j$ and $b^h_j$, for $h=1,\ldots,n$ and $j=0,1,\ldots,m$.
\end{proposition}
\begin{proof}
We need to prove that the function $Z$ introduced in \eqref{coordinates} is a diffeomorphism between $\mathcal{G}$ and its image. We denote by $\partial_{\hat{r}}Z:T_{\hat{r}}\mathcal{G}\rightarrow \R^n$ the Fr\'echet derivative of $Z$, where $T_{\hat{r}}\mathcal{G}$ is the tangent space of $\G$ at $\hat{r}$. Since $\mathcal{G}\subseteq\hat{\mathcal{H}}$ has dimension $n$, $T_{\hat{r}}\mathcal{G}$ is an $n$-dimensional subspace of $\hat{\mathcal{H}}$. Let us consider a basis $(\hat{e}_1,\dots,\hat{e}_n)$ for $T_{\hat{r}}\mathcal{G}$, where we adopt the notation $\hat{e}_h=(e^{0}_h,\dots,e^{2m}_h)\in\hat{\mathcal{H}}$, for $h= 1,\dots,n$. Then, for a generic element $v\in T_{\hat{r}}\G$, there exists $\gamma = (\gamma_1,\dots,\gamma_n)^{\top}\in\R^n$ such that $v = \sum_{h=1}^n\gamma_h\hat{e}_h$. 
By linearity, the Fr\'echet derivative of $Z$ can be written as 
\[
\partial_{\hat{r}}Z\cdot v= \sum_{h=1}^n\gamma_h (\partial_{\hat{r}}Z\cdot \hat{e}_h) = K_n(x)\gamma,
\] 
where
\[
K_n(x) := 
\begin{pmatrix}
\alpha^1\cdot\hat{e}_1(x_1)&\cdots&\alpha^1\cdot\hat{e}_n(x_1)\\ 
\vdots&\ddots&\vdots\\ 
\alpha^n\cdot\hat{e}_1(x_n)&\cdots&\alpha^n\cdot\hat{e}_n(x_n)
\end{pmatrix},
\]
and 
\begin{equation}
\label{alpha-bom}
\alpha^h = (a^h_0,a^h_1,\dots,a^h_m,b^h_1,\dots,b^h_m)\in\mathbb{R}^{2m+1}, 
\qquad\text{ for } h= 1,\dots,n.
\end{equation}
Therefore, the function $Z$ is a local system of coordinates if $K_n(x)$ is invertible. 
Since $\mathcal{H}$ is composed by analytic functions (see, e.g., \cite[Proposition 4.2]{bjork2001existence}), a non-null element of $\mathcal{H}$ has isolated zeroes and we can find a discrete subset $N\subset\R_+$ such that $\hat{e}_h(x)\neq 0$ for all $h= 1,\dots,n$ and $x\notin N$. Since $\mathrm{det}(K_n(x))$ is a polynomial function of $(\hat{e}_h(x))_{h=1,\dots,n}$, we can always find a set of vectors $(\alpha^h)_{h = 1,\ldots,n}$, defined as in \eqref{alpha-bom} and dependent on $x$, such that $\mathrm{det}(K_n(x)) \neq0$ for every $x= (x_1,\dots,x_n)$, where $x$ is arbitrarily chosen in $\R^n_+$ except for a discrete set.
\end{proof}

\section{An alternative formulation of invariance} 
\label{section - alternative invariance}

In the previous sections, the log-spread processes $(Y^j)_{t\geq0}$, for $j=1,\ldots,m$, have been treated jointly with the forward rate curves. However, since the log-spread processes are inherently finite-dimensional, we can aim at a formulation of invariance that treats the log-spreads differently from the forward rate components.
More specifically, in this section we aim at understanding if the log-spread processes can be directly included in the state process $(Z_t)_{t\geq0}$ of an FDR.

\subsection{The role of log-spreads in the consistency problem}

Theorem \ref{invariance-thm} characterizes the consistency between a multi-curve interest rate model $\mathcal{M}$ and a manifold $\mathcal{G}$, which provides a joint representation of forward rates and log-spreads. As explained in Remark \ref{rem:joint_spreads}, in a multi-curve setting consistency cannot be addressed for the forward rate components alone. Motivated by this remark, in this section we investigate an alternative formulation of invariance, which directly includes the log-spreads in the state variables. 
More specifically, adopting the notation $\hat{r}_t= (r_t,Y_t)\in\hat{\mathcal{H}} $, we investigate under which conditions there exist a stopping time $\tau>0$ a.s., a process $(Z_t)_{t\geq0}$ taking values in $\mathcal{Z}\subseteq\R^n$ and a function $\bar{G}:\R^{m}\times\mathcal{Z}\rightarrow \nH^{m+1}$ such that
\begin{equation}
\label{invariance v2}
r_t(x) 	= \bar{G}(Y_t, Z_t,x)\text{ a.s.},
\qquad \text{ for all }x\in\R_+\text{ and }t\in[0,\tau). 
\end{equation}
This is equivalent to the notion of $\hat{r}$-invariance (Definition \ref{r invariance def}), for the immersion $G$ given by
\[
G(y,z) := (\bar{G}(y,z),y).
\]

We can rewrite as follows the Stratonovich dynamics \eqref{forward-rate system} of the joint process $\hat{r}$, distinguishing explicitly the forward rate components from the log-spread processes:
\begin{equation}
\label{fr+sp dynamics}
d\hat{r}_t= 
\begin{pmatrix}
dr_t\\ 
dY_t
\end{pmatrix}
=
\begin{pmatrix}
\mu(r_t, Y_t) dt + \sigma(r_t, Y_t) \circ dW_t\\ 
\gamma(r_t,Y_t) dt + \beta(r_t,Y_t)\circ dW_t
\end{pmatrix}.
\end{equation}

We define as follows the alternative notion of invariance considered in this section.

\begin{defin}
\label{y invariance def}
A parameterized family $\bar{G}$ is $(r,y)$-invariant under the action of  $\hat{r}=(r,y)$ if, for every initial point $(r_0,Y_0)$, there exists an a.s. strictly positive stopping time $\tau(r_0,Y_0)$ and a stochastic process $(Z_t)_{t\geq0}$, taking values in $\mathcal{Z}$ and with Stratonovich dynamics \eqref{eq:Stratonovich_Z}, such that condition \eqref{invariance v2} holds.
\end{defin}

Similarly to the equivalence between Definition \ref{invariance} and Definition \ref{r invariance def}, it can be easily shown that Definition \ref{y invariance def} is equivalent to Definition \ref{invariance} with $G:= (\bG,\mathbb{I}_m)$. As a consequence, the following result can be proved analogously to Theorem \ref{invariance-thm}, characterizing the validity of \eqref{invariance v2}.

\begin{proposition} 
Let $\bG:\R^{m}\times\mathcal{Z}\rightarrow\nH^{m+1}$ be a parameterized family such that $G:= (\bG,\mathbb{I}_m)$ satisfies Assumption \ref{assumption - immersion condition}. Then $\bar{G}$ is $(r,y)$-invariant under the action of  $\hat{r}=(r,y)$ if and only if, for every $(z,y)\in \mathcal{Z} \times \R^m$:
\[
\begin{cases}
\mu((\bG(y,z),y)) = \mathrm{Im}[\partial_z\bG(y,z)] + \partial_y\bG(y,z)\gamma(y,z),\\ 
\sigma_i((\bG(y,z),y)) = \mathrm{Im}[\partial_z\bG(y,z)] + \partial_y\bG(y,z)\beta_i(y,z),\quad \text{ for all }i = 1,\dots,d,
\end{cases}
\]
where $\partial_z\bG$ and $\partial_y\bG$ denote the Fr\'echet differentials of $\bG$ with respect to $z$ and $y$, respectively.
\end{proposition}

\subsection{Existence of FDRs in the form of Definition \ref{y invariance def}}

In this subsection, we address the issue of the existence of FDRs that can be realized in the form \eqref{invariance v2}. We start by observing that, if a multi-curve interest rate model $\mathcal{M}$ admits an FDR (in the sense of Definition \ref{fdr def}), then it is always possible to construct another FDR in the form \eqref{invariance v2}. Indeed, let us consider an FDR given by a function $G= (\bG,\wG)$ and a finite-dimensional process $(Z_t)_{t\geq0}$ such that $(r_t,Y_t) = (\bG(Z_t),\wG(Z_t))$. We can obviously obtain an FDR of the form \eqref{invariance v2} by considering $(r_t,Y_t) = (\bG(Z_t),Y_t)$ and the joint process $((Y_t,Z_t))_{t\geq0}$ with dynamics
\begin{displaymath}
d\begin{pmatrix}
Y_t\\ 
Z_t
\end{pmatrix}
= 
\begin{pmatrix}
\gamma(\bG(Z_t),Y_t) dt + \beta(\bG(Z_t),Y_t)\circ dW_t\\ 
a( Z_t) dt +b(Z_t) \circ dW_t\\ 
\end{pmatrix}.
\end{displaymath}

However, the above strategy of adding the whole vector $Y_t$ of log-spreads to the state variables $Z_t$ of an FDR might yield a joint process $((Y_t,Z_t))_{t\geq0}$ with redundant components. It is therefore of interest to determine conditions under which a more parsimonious FDR can be found (to this effect, compare also the two approaches described in Section \ref{subsection - HL NS} and see Remark \ref{parsimonious strategy}).

We start by noting that the embedding $G$ introduced in  \eqref{parameterized function G - finite-dimensional dimension} can be decomposed as 
\[
G(z) =( \bar{G}(z), \widetilde{G}(z)),
\] 
for $z\in\R^n$, where $\bar{G}$ takes values in $\mathcal{H}^{m+1}$ and $\widetilde{G}$ in $\R^m$. The following assumption ensures the possibility of replacing some components of the state vector $z$ by log-spreads.

\begin{condition}
\label{condition ywz}
There exists a subspace $\R^p\subseteq\mathbb{R}^m$, for $p\leq m$, that is diffeomorphic to a subspace of the state space $\mathcal{Z}\subseteq\R^n$ through the function $\widetilde{G}$. 
\end{condition}

Under Condition \ref{condition ywz}, we denote by $w$ the elements of the diffeormophic subspace of the state space $\mathcal{Z}$, while $\tilde{y}$ denotes a generic element of $\R^p$. For $z\in\mathcal{Z}$, we write $z = (\bar{z},w)\in\mathbb{R}^{n-p}\times\mathbb{R}^p$. 
Without loss of generality, we assume that $\tilde{y}$ represents the first $p$ components of the log-spread vector $y$, while $\bar{y}$ denotes the last $m-p$ components, so that $y=(\tilde{y},\bar{y})$.
Writing $\widetilde{G}= (\widetilde{G}_p,\widetilde{G}_{m-p})$, invariance implies that $\tilde{y} = \wG_{p}(\bar{z},w) =: \wG_p[\bar{z}](w)$ is invertible in $w$, for every  $\bar{z}\in\R^{n-p}$. Hence, there exists an inverse mapping $\widetilde{G}_{p}[\bar{z}]^{-1}$ such that $w = \widetilde{G}_{p}[\bar{z}]^{-1}(\tilde{y})$. 
Let us remark that Condition \ref{condition ywz} does not impose additional requirements on the model. Indeed, if a diffeomorphic subspace does not exist, then we can just assume that Condition \ref{condition ywz} holds with $p=0$.

To construct an FDR under Condition \eqref{condition ywz}, we first recall that FDRs are obtained by a set of generators of the Lie algebra $\{\hat{\mu},\hat{\sigma}_1,\ldots,\hat{\sigma}_d\}_{\mathrm{LA}}$ (see Section \ref{section - FDRs}). We denote the generators by $\xi_1,\ldots,\xi_n$. To construct an FDR including the log-spreads among the state variables, it suffices to add the last $m-p$ elements of the canonical basis of $\hat{\mathcal{H}}$ to the set of generators $(\xi_1,\ldots,\xi_n)$. 
Using notation \eqref{N constant direction vol}, we denote by $\gamma_k$ the $k$-th vector of the canonical basis, for $k=p+1,\ldots,m$. In line with Theorem \ref{thm:FDR}, a parsimonious FDR that includes the vector of log-spreads among the state variables can then be constructed under the following assumption.

\begin{assumption}
\label{finite aug lie alg}
The Lie algebra generated by $\xi_1,\dots,\xi_n,\gamma_{p+1},\dots,\gamma_{m}$ is finite-dimensional. 
\end{assumption}

Note that Assumption \ref{finite aug lie alg} does not necessarily hold even if $\{\xi_1,\ldots,\xi_n\}_{\mathrm{LA}}$ is finite-dimensional.
We now aim at obtaining necessary and sufficient conditions for the validity of Assumption  \ref{finite aug lie alg}. As a preliminary, we recall the notion of multi-index (see, e.g., \cite[Definition 7.4]{bjork2004geometry}).

\begin{defin}
A multi-index $\alpha\in\mathbb{Z}^{k}_+$ is any vector of dimension $k$ with nonnegative integer elements. For a multi-index $\alpha = (\alpha_1,\dots,\alpha_k)$, the differential operator $\partial^{\alpha}_{y}$ is defined as
\begin{displaymath}
\partial^{\alpha}_{y} := \frac{\partial^{\alpha_1}}{\partial y_{p+1}^{\alpha_1}}\frac{\partial^{\alpha_2}}{\partial y_{p+2}^{\alpha_2}}\cdots\frac{\partial^{\alpha_k}}{\partial y_k^{\alpha_k}}.
\end{displaymath}
\end{defin}

\begin{proposition}
If the Lie algebra
\[
\mathcal{N}:=\{\xi_1,\dots,\xi_n,\gamma_{p+1},\dots,\gamma_{m}\}_{\mathrm{LA}}
\]
is finite-dimensional, then
\begin{equation}
\label{necessary conditions}
\begin{cases}
\mathrm{dim}[\mathrm{span}\{\partial^{\alpha}_{y}\hat{\mu}(r,{y});\ \alpha\in\mathbb{Z}^{m-p}_+\}]<+\infty,\\ 
\mathrm{dim}[\mathrm{span}\{\partial^{\alpha}_{y}\hat{\sigma}_i(r,{y});\ \alpha\in\mathbb{Z}^{m-p}_+\}]<+\infty,\quad\text{ for all } i=1,\dots,d.
\end{cases}
\end{equation}
Conversely, if $\gamma_k$ commutes with $\hat{\mu}$ and $\hat{\sigma}_i$ for every $k = p+1,\dots,m$ and $i=1,\ldots,d$, then the Lie algebra $\mathcal{N}$ is finite-dimensional.
\end{proposition}
\begin{proof}
The Lie algebra $\mathcal{N}$ contains all Lie brackets of the form
\[
[\xi_i,\gamma_k] = \partial_{\hat{r}} \xi_i\,\gamma_k - \partial_{\hat{r}}\gamma_k\,\xi_i = \partial_{y_k} \xi_i, \qquad \text{ for all }k = p+1,\dots,m. 
\]
Hence, all the differentials $\partial^{\alpha}_y\xi_i$ are contained in $\mathcal{N}$, for every multi-index $\alpha$. We first notice that $\mathcal{N}$ coincides with $\{\hat{\mu},\hat{\sigma}_1,\dots,\hat{\sigma}_d,\gamma_{p+1},\dots,\gamma_{m}\}_{\mathrm{LA}}$ and we observe that the vectors $\gamma_k$ commute with each other, so that their Lie bracket is null. Therefore, in order to have that $\mathrm{dim}[\mathcal{N}]<+\infty$, it is necessary that $\partial^{\alpha}_y\mu$ and $\partial^{\alpha}_y\sigma_i$, for $i = 1,\dots,d$, do not generate an infinite-dimensional distribution, for every $\alpha \in \mathbb{Z}^{m-p}_+$. This implies the necessity of condition \eqref{necessary conditions}.

On the other hand, assuming that $\gamma_k$ commutes with $\hat{\mu}$ and $\hat{\sigma}_i$ is equivalent to requiring that
\[
[ \hat{\mu},\gamma_k ] = 0
\qquad\text{ and }\qquad
[ \hat{\sigma}_i,\gamma_k ] = 0,\quad\text{ for all } i =1,\dots,d,
\]
for every $k = p+1,\dots,m$. By the Jacobi identity, successive Lie brackets commute with $\gamma_k$: indeeed, it holds that $[[\hat{\mu},\hat{\sigma}_i],\gamma_k] = - [[\hat{\sigma}_i,\gamma_k],\hat{\mu}] - [[\gamma_k,\hat{\mu}],\hat{\sigma}_i]= 0$. This implies that the commutativity of $\gamma_k$ with $\hat{\mu}$ and $\hat{\sigma}_i$ is a sufficient condition for $\mathcal{N}$ to be finite-dimensional.
\end{proof}

\begin{remark}[Constant direction volatility models]
In Section \ref{subsection - example constant direction volatility}, to prove the existence of FDRs for constant direction volatility models, we studied the Lie algebra $\mathcal{L}^1$, defined in \eqref{L1 cdv}, that is larger than the Lie algebra $\{\hat{\mu},\hat{\sigma}_1,\ldots,\hat{\sigma}_d\}_{\mathrm{LA}}$. In Proposition \ref{sufficient condition constant direction theorem}, we provided conditions ensuring that $\mathrm{dim}[\mathcal{L}^1]<+\infty$. Under these conditions, we can construct FDRs for which the log-spreads are included in the state variables. Indeed, Assumption \ref{finite aug lie alg} holds under the conditions of Proposition \ref{sufficient condition constant direction theorem}, since $\mathcal{L}^1$ already contains all the vector fields $(\gamma_1,\ldots,\gamma_m)$. 
\end{remark}

\section{Calibration of finite-dimensional realizations to market data} 
\label{section - calibration algorithm}

In this section, we study the calibration of a multi-curve interest rate model to market data, relying on the theoretical results presented in Section \ref{section - FDRs}. More specifically, we consider a model admitting FDRs and depending on a parameter vector $\theta$. The calibration procedure aims at determining the parameter vector $\theta^*$ that achieves the best fit to the market data. As a result of this procedure, the FDR associated to $\theta^*$ will represent the submanifold of $\hat{\mathcal{H}}$ that gives the best representation (in terms of mean-squared error) of the data under analysis. 

We consider a three-curve Hull-White model driven by a one-dimensional Brownian motion, for simplicity of presentation. The model is fully specified by the volatility
\[
\hat{\sigma}(\hat{r}) = 
( \sigma^0e^{-a^0x},	\sigma^1e^{-a^1x}, \sigma^2e^{-a^2x}, \beta^1, \beta^2 ),
\]
where $\beta^j,a^j,\sigma^j$ are positive constants, for $j=0,1,2$. 
This specification is a special case of the constant volatility models studied in Section \ref{subsection - example constant volatility} and satisfies the conditions of Proposition  \ref{characterization FDRs constant vol}. Indeed, the volatility functions are QE, since $(\textbf{F}+a^j)\sigma^j= 0$, for all $j=0,1,2$. 
For the model under analysis, the parameter vector is given by $\theta=(a^0,\sigma^0,a^1,\sigma^1,a^2,\sigma^2,\beta^1,\beta^2)$.
The algorithm described below for calibrating $\theta$ is based on the works of \cite{angelini2002consistent}, \cite{angelini2005consistent}  and \cite{slinko2010finite}.

\subsection{Construction of the FDRs}
\label{subsection - HL fdrs}

We first study the vector fields generating $\mathcal{L}:= \{\hat{\mu},\hat{\sigma}\}_{\mathrm{LA}}$. The annihilator polynomial of the forward rate components of the volatility $\hat{\sigma}$ is given by
\begin{displaymath}
M(\gamma) := (\gamma+a^0)(\gamma+a^1)(\gamma+a^2) = \gamma^3 + \alpha_3\gamma^2 + \alpha_2\gamma^1 + \alpha_1.
\end{displaymath}
By equation \eqref{dimension of L}, $\mathrm{dim}(\mathcal{L})= 5$ is the dimension of the FDRs. Accordingly, we consider a state vector in $\mathbb{R}^5$, denoted by $z=(z^0,z^0_1,z^1_1,z^2_1,z^3_1)$, in line with notation \eqref{notation}. 
By following the arguments of Section \ref{subsection - example constant volatility}, we can construct the tangential manifold determined by the composition of the integral curves of the generators of $\{\hat{\mu},\hat{\sigma}\}_{\mathrm{LA}}$. The generators are $\hat{\mu}$, $\hat{\sigma}$ and 
\begin{small}
\begin{align*}
\nu^1 	=  \begin{pmatrix}	-a^0\sigma^0e^{-a^0x}\\ 	-a^1\sigma^1e^{-a^1x}\\ 	-a^2\sigma^2e^{-a^2x}\\ 	\sigma^0 - \sigma^1\\ 	\sigma^0 - \sigma^2 \end{pmatrix},\quad
\nu^2 	=  \begin{pmatrix}	(a^0)^2\sigma^0e^{-a^0x}\\ 	(a^1)^2\sigma^1e^{-a^1x}\\ 	(a^2)^2\sigma^2e^{-a^2x}\\ 	-a^0\sigma^0 +a^1\sigma^1\\ 	-a^0\sigma^0 +a^2\sigma^2	 \end{pmatrix},\quad
\nu^3 	=  \begin{pmatrix}	-(a^0)^3\sigma^0e^{-a^0x}\\ 	-(a^1)^3\sigma^1e^{-a^1x}\\ 	-(a^2)^3\sigma^2e^{-a^2x}\\ 	(a^0)^2\sigma^0 -(a^1)^2\sigma^1\\ 	(a^0)^2\sigma^0 -(a^2)^2\sigma^2	 \end{pmatrix}.
\end{align*}
\end{small}
The composition of the integral curves of these vector fields yields the tangential manifold $G= (G^0,G^1,G^2,G^3,G^4)$, that can be explicitly computed as follows:
\begin{equation}
\begin{split}
\label{fdr fr HL}
G^j(z,x) 			&= r^M_j(x+z^0) +\sigma^je^{-a^jx}(z^0_1-a^jz^1_1+(a^j)^2z^2_1-(a^j)^3z^3_1)\\ 
				&\qquad +\frac{1}{2}\Bigl{(}\frac{\sigma^j}{a^j}\Bigr{)}^2e^{-2a^jx}\bigl(e^{-2a^jz^0}-1\bigr)- \frac{\sigma^j}{a^j}\left(\frac{\sigma^j}{a^j}-\delta^j_0\beta^j\right)e^{-a^jx}\bigl(e^{-a^jz^0}-1\bigr),\quad j = 0,1,2,\\ 
G^{2+j}(z)			&=(\sigma^0-\sigma^j)z^1_1+(-a^0\sigma^0+a^j\sigma^j)z^2_1+((a^j)^2\sigma^0-(a^j)^2\sigma^j)z^3_1+\beta^jz^0_1+y^M_j\\ 
				&\qquad +\int_0^{z^0}(r^M_0(s)-r^M_j(s))ds+\frac{1}{2}\Bigl{(}\frac{\sigma^0}{a^0}\Bigr{)}^2\Bigl(z^0-\frac{2}{a^0}\bigl(1-e^{-a^0z^0}\bigr)+\frac{1}{2a^0}\bigl(1-e^{-2a^0z^0}\bigr)\Bigr)\\ 
				&\qquad-\frac{1}{2}\Bigl{(}\frac{\sigma^j}{a^j}\Bigr{)}^2\Bigl{(}z^0-\frac{2}{a^j}\bigl(1-e^{-a^jz^0}\bigr)+\frac{1}{2a^j}\bigl{(}1-e^{-2a^jz^0}\bigr{)}\Bigr{)}+\frac{\sigma^j}{a^j}\beta^j\Bigl{(}z^0-\frac{1}{a^j}\bigl{(}1-e^{-a^jz^0}\bigr{)}\Bigr{)}\\ 
				&\qquad-\frac{1}{2}(\beta^j)^2z^0,
\qquad\qquad j = 1,2.
\end{split}
\end{equation}
The state process $(Z_t)_{t\geq0}$ is the solution to the SDE $dZ_t = A(Z_t)dt+B(Z_t)\circ dW_t$, where $A,B$ are vector fields on $\mathbb{R}^5$, respectively defined by the coefficients $a$ and $b$ introduced in system \eqref{constant vol - Z coefficients}. 
By \eqref{constant vol - Z coefficients}, the first component of $Z_t$ is simply given by $Z^0_t = t$, so that we write $Z_t = (t,Z_{1t})$. While  the dynamics of $(Z_{1t})_{t\geq0}$ can be derived explicitly, they are not needed in the following. 

\subsection{Initial families}
\label{init families}

The FDRs depend on the initial term structures. For the representation of the initial term structures $r^M_j$ of the forward rates, for $j=0,1,2$, we adopt the widely used Nelson-Siegel family (see \cite{nelson1987parsimonious}), given here in the following form:
\begin{equation}
\label{NS initial}
r^M_j(y,x)= y_0 + y_1e^{-a^j x} + y_2 x e^{-a^j x} 
=: {M}_j^0(x;a^j)\cdot {y},\quad\text{ for } j =0,1,2,
\end{equation}
where $y= (y_0,y_1,y_2)$ and ${M}_j^0(x;a^j):= ( 1 ,\ e^{-a^jx} ,\  xe^{-a^j} )$. Observe that in \eqref{NS initial} we directly use the parameter $a^j$ in the exponents.
The Nelson-Siegel family has been also adopted in \cite{slinko2010finite}, thus facilitating the comparison of our methodology with that work.
As a consequence of \eqref{NS initial}, the initial families of the forward rate components of the model depend linearly on the common vector $y$, the only difference being in the exponent $a^j$ that is specific to each forward rate. 

\subsection{The calibration procedure}
\label{cal proc}

We consider market data at daily frequency $\{t_0,\dots,t_N\}$. For each day $t\in\{t_0,\dots,t_N\}$, we extract from market data (by means of standard bootstrapping techniques, see below for more details) risk-free ZCB prices, prices of fictitious ZCBs associated to the risk-sensitive rates, for a set of maturities $\bar{x}:=\{x_1,\dots,x_n\}$, and the log-spreads:
\begin{equation}
\label{mk data}
\text{MK\_data}_t:=
\bigl(
B^0_t(x_1), \dots,B^0_t(x_n),B^1_t(x_1),\dots,B^2_t(x_n),Y^1_t,Y^2_t
\bigr)\in\R^{3n+2}.
\end{equation}

For every $t\in\{t_0,\dots,t_N\}$, we minimize the squared error between   the yields computed on the market data \eqref{mk data} and the yields generated by the FDR \eqref{fdr fr HL}, for all maturities in $\bar{x}$ and for all tenors. 
We denote by $G^j(\cdot; t,z_1, y;\theta)$ the FDR, highlighting the dependence on the parameter vector $\theta$ to be estimated. The argument $z_1$ represents the realization of the components $Z_{1t}$ of the state process, which need to be estimated at each date.  
We denote by $\text{Res}_{t}({z}_1,y;\theta) \in \R^{3n+2}$ the residual at date $t\in\{t_0,\ldots,t_N\}$, given as follows:
\begin{equation}
\label{residual formula}
\mathrm{Res}_{t}(z_1,y;\theta) := 
\begin{pmatrix}
\frac{1}{x_1}\bigl{(}-\int_0^{x_1}G^0(u; t,z_1,y;\theta)du - \log{B^0_t(x_1)}\bigr{)}\\
\vdots\\ 
\frac{1}{x_n}\bigl{(}-\int_0^{x_n}G^0(u;t,z_1,y;\theta)du - \log{B^0_t(x_n)}\bigr{)}\\
\frac{1}{x_1}\bigl{(}-\int_0^{x_1}G^1(u;t,z_1,y;\theta)du - \log{B^1_t(x_1)}\bigr{)}\\
\vdots\\ 
\frac{1}{x_n}\bigl{(}-\int_0^{x_n}G^2(u;t,z_1,y;\theta)du - \log{B^2_t(x_n)}\bigr{)}\\
G^{3}(t,z_1,y;\theta)-Y^1_t\\ 
G^{4}(t,z_1,y;\theta)-Y^2_t
\end{pmatrix}.
\end{equation}

The properties of the FDR derived in \eqref{fdr fr HL} and the choice of the initial family made in \eqref{NS initial} imply that the yields generated by the FDR are affine in $(z_1,y)$. Similarly as in \cite{angelini2002consistent}, this represents a significant advantage in the calibration algorithm, which is structured as follows:
\begin{enumerate}[label=\textbf{P.\arabic*}]
\item \label{cal procedure - step 1} 
For each $t\in\{t_0,\dots,t_N\}$, we minimize $\mathrm{Res}_{t}(z_1,y;\theta)$ with respect to the parameters $(z_1,y)$. Exploiting the affine structure of the yields, the SVD algorithm can be used to obtain $z_1(t,\theta)$ and $y(t,\theta)$, depending on the parameter vector $\theta$:
\[
\bigl(z_1(t,\theta),y(t,\theta)\bigr) := \mathrm{arg}\min_{(z_1,y)}\bigl\|\text{Res}_{t}(z_1,y;\theta)\bigr\|.
\]
\item \label{cal procedure - step 2} 
The (time-independent) parameter vector $\theta$ is estimated by minimizing the sum of the squared residuals obtained in the previous step along the entire dataset:
\[
\theta^* := \mathrm{arg}\min_{\theta}\sum_{h = 0}^N\bigl|\text{Res}_{t_h}(z_1(t_h,\theta),y(t_h,\theta);\theta)\bigr|^2.
\]
To compute the minimizer $\theta^*$, we adopt a reflective trust-region algorithm. 
\end{enumerate}

\subsection{Market data}
\label{real mk data}

For the risk-free rate we rely on OIS rates, while as risk-sensitive rates we consider 3M and 6M Euribor rates, corresponding to the most liquidly traded tenors. The dataset used in our analysis is given by daily market quotes from $10/08/2016$ until $19/11/2021$. Table \ref{table: mk data} presents a snapshot of the market instruments included in our dataset.
\begin{table}[h!]
\begin{small}
\begin{tabular}{lll}
\toprule
Interest rate curve	&	Market instrument	&	Quoted maturities\\ 
\midrule
Risk-free curve		&	OIS				&	1W - 2W - 3W - 1M - 2M - 3M - 4M - 5M - 6M - 7M - \\ 
				&					&	8M - 9M - 10M - 11M - 1Y - 15M - 18M - 21M - 2Y - \\
				&					&	3Y - 4Y - 5Y - 6Y - 7Y - 8Y - 9Y - 10Y\\ 
\hline
3M		&	FRA				&	1Mx4M - 2Mx5M - 3Mx6M - 4Mx7M - 5Mx8M - 6Mx9M - \\ 
				&					&	7Mx10M - 8Mx11M - 9Mx12M\\ 
				&	IRS				&	18M - 2Y - 3Y - 4Y - 5Y - 6Y - 7Y - 8Y - 9Y - 10Y\\ 
\hline
6M		&	FRA				&	1M+7M - 2Mx8M - 3Mx9M - 4Mx10M - 5Mx11M - 6Mx12M - \\ 
				&					&	9Mx15M - 12Mx18M\\ 
				&	IRS				&	2Y - 3Y - 4Y - 5Y - 6Y - 7Y - 8Y - 9Y - 10Y\\
\bottomrule
\end{tabular}
\end{small}
\caption{Summary of market data.}
\label{table: mk data}
\end{table}

On the basis of this market data, we compute the risk-free and risk-sensitive term structures, making use of the bootstrapping technique described in \cite{gerhart2020empirical}. The resulting term structures include the following maturities: $\{$1M, 2M, 3M, 4M, 5M, 6M, 9M, 1Y, 2Y, \dots, 10Y$\}$.

\subsection{Calibration results}
\label{results cal proc}

To assess the performance of the calibration algorithm, we compare the market data with the calibrated parameterized family at the end of the considered time window. More precisely, the calibrated parameterized family makes use of the parameter vector $\theta^*$ calibrated as explained in Section \ref{cal proc}, while the time-dependent parameters $(z,y)$ are estimated on the basis of the market data at the end of the time window.

By means of a stability analysis (see Section \ref{subsubsection - stability of theta} below), we have determined that a time window of four months yields the most stable results. We therefore consider a time window of four months, starting at 01/04/2021. 
The initial guesses $\theta_0$ for the parameters $\theta$ are given in Table \ref{theta 0}. The parameters $a^j,\beta^j$ in Table \ref{theta 0} are randomly chosen  in the interval $[0,1]$, while $\sigma^j$ is randomly chosen in the interval $[0,0.1]$.
The calibrated values are reported in Table \ref{theta est}.

\begin{table}[ht]
\centering
\begin{tabular}{ | c | c | c | c | c | c | c | c | }
 \hline
		 	& $\sigma$   			& $a$				& $\beta$  \\
 \hline
OIS 			& $\sigma^0 = 0.00285941$	&$a^0 = 0.53041117$  			& /  \\
 \hline
 Libor - 6M 	& $\sigma^1 =0.09546952	$	&$a^1 =  0.66253001$			& $\beta^1 =0.41734616$   \\
 \hline
  Libor - 6M 	& $\sigma^2 =0.09083773$ 	&$a^2 =0.65812121$ 			&$\beta^2 = 0.82477578$ \\
 \hline
\end{tabular}
\caption{\footnotesize{Initial guesses of the model parameters.}}
\label{theta 0}
\end{table}

\begin{table}[ht]
\centering
\begin{tabular}{ | c | c | c | c | c | c | c | c | }
 \hline
		 	& $\sigma$   			& $a$				& $\beta$  \\
 \hline
OIS 			& $\sigma^0 = 0.1643$	&$a^0 = 0.3719$  			& /  \\
 \hline
 Libor - 6M 	& $\sigma^1 =0.1590	$	&$a^1 =  0.3721$			& $\beta^1 =0.4814$   \\
 \hline
  Libor - 6M 	& $\sigma^2 =0.1598$ 	&$a^2 = 0.3727$ 			&$\beta^2 = 0.8825$ \\
 \hline
\end{tabular}
\caption{\footnotesize{Calibrated values of the model parameters.}}
\label{theta est}
\end{table}

The quality of the fit of the calibrated parameterized family to the market data at the end of the considered time window is illustrated in Figure \ref{results yields}. Apart from the shortest maturities, the quality of the fit appears satisfactory.
We can also notice that, due to the unusual monetary policy conditions of 2021, the yields are negative for all maturities. 
Since the spreads are spot processes, we can compare the calibrated log-spreads with respect to the log-spreads obtained from market data on the whole time window. This comparison is illustrated in Figure \ref{results spreads}.

\begin{figure}[ht]
\centering
\includegraphics[scale = 0.55]{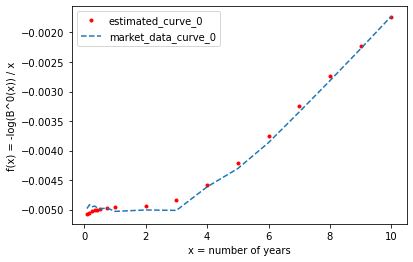}\hfil
\includegraphics[scale = 0.55]{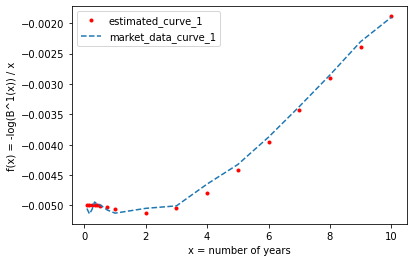}\hfil
\includegraphics[scale = 0.55]{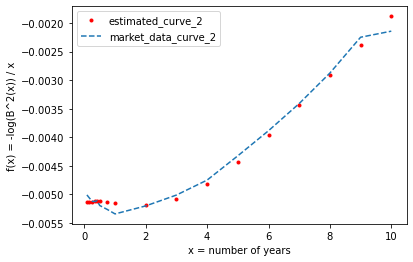}\\
\caption{\footnotesize{Comparison between market yield and the calibrated yield curve at the end of time window. Top panel: risk-free curve; central panel: 3M curve; bottom panel: 6M curve.}}
\label{results yields}
\end{figure}

\begin{figure}[ht]
\centering
\includegraphics[scale = 0.55]{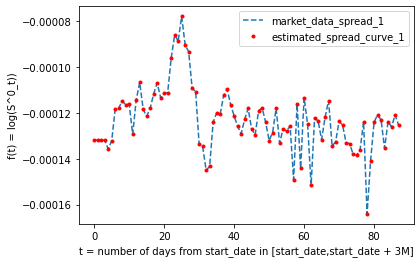}\hfil
\includegraphics[scale = 0.55]{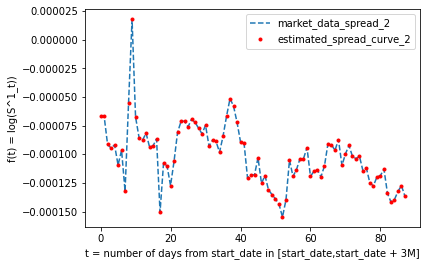}
\caption{\footnotesize{Comparison between market data and the calibrated log-spread for the whole time series. Top panel: 3M log-spread; bottom panel: 6M log-spread.}}
\label{results spreads}
\end{figure}

In Table \ref{table errors}, the relative errors are given. The errors have been computed as follows:
\begin{enumerate}
\item if $G^j(\bar{x})$ is the calibrated $j$-th yield curve and $M^j(\bar{x})$ is the $j$-th yield curve obtained from market data, both considered at the end of the time window, then 
\begin{equation}
\label{error yield}
\text{err}_{\text{yield curve}} := \frac{||G^j(\bar{x}) - M^j(\bar{x})||_n}{||M^j(\bar{x})||_n};
\end{equation}
 \item if $(Y^j_t)_{t\in\{t_0,\dots,t_N\}}$ is the estimated value of the $j$-th log-spread over the entire time window $\{t_0,\dots,t_N\}$ used for the calibration and $(M^j_t)_{t\in\{t_0,\dots,t_N\}}$ is the market value of the $j$-th log-spread in the considered time window, the relative error is
\begin{equation}
\label{error spread}
\text{err}_{\text{spread curve}} := \frac{||S^j - M^j||}{||M^j||}= \frac{\sqrt{\sum_{h = 0}^N(Y^j_{t_h} - M^j_{t_h})^2}}{\sqrt{\sum_{h = 0}^N(M^j_{t_h})^2}}.
\end{equation}
\end{enumerate}

\begin{table}[ht]
\begin{centering}
\begin{tabular}{llll}
\toprule
{} 			&OIS 	&3M 			&                     6M \\
\midrule
yields 		&0.01917	&0.01705 		&0.02385 \\
spread 		&      - 	& 6.92929e-07 	&8.491172e-07 \\
\bottomrule
\end{tabular}
\caption{\footnotesize{Relative errors.}}
\label{table errors}
\end{centering}
\end{table}

\subsubsection{Stability with respect to the length of the time window}
\label{subsection - length error}

When calibrating the model, it is essential to choose in a suitable way the length of the time window. Indeed, a too short time window does not convey sufficient information for a reliable estimation. On the contrary, a too long time window can also be problematic, since consistency is a local property by definition.

In Table \ref{error length} we report the relative errors obtained using time windows of different lengths. 
For the yield curves, we report the error as defined in equation \eqref{error yield}. For the spreads, we report the relative error between the estimated spread and the market data at the end of the time window. The calibration procedure is always initialized at $\theta_0$ as given in Table \ref{theta 0}. We consider time series of market data ending in 30/07/2021, with lengths 1M, 2M, 3M, 4M, 5M and 6M.
This analysis reveals that, for the dataset under consideration, the time length that achieves the best performance is four months, as shown in Table \ref{error length}.


\begin{table}[ht]
\center
\small
\begin{tabular}{ |c||c|c|c||c|c| }
\hline
\textbf{length of $\{t_0,\dots,t_N\}$} &\multicolumn{3}{c||}{\textbf{Yield curve}} & \multicolumn{2}{c|}{\textbf{Spread curve}} \\
\hline
{months} &      RFRs &  Euribor 3M &  Euribor 6M &     Spread 3M &     Spread 6M \\
\hline
1 &  0.0520 &    0.0524 &    0.0368 &  2.5736e-07 &  2.4474e-07 \\
2 &  0.0557 &    0.0489 &    0.0371 &  5.6650e-08 &  1.6761e-10 \\
3 &  0.0213 &    0.0218 &    0.0387 &  4.4926e-07 &  3.8488e-07 \\
4 &  0.0191 &    0.0171 &    0.0239 &  1.1579e-06 &  9.6982e-07 \\
5 &  0.0159 &    0.0163 &    0.0334 &  3.1638e-08 &  3.4686e-08 \\
6 &  0.0213 &    0.0214 &    0.0393 &  1.2390e-07 &  6.3183e-08 \\
\hline
\end{tabular}
\caption{\footnotesize{Relative error as a function of the length (in months) of the time window.}}
\label{error length}
\end{table}

\subsubsection{Stability of the time-independent parameters}
\label{subsubsection - stability of theta} 

In practice, the stability of the calibrated parameters represents an important property. We test this stability through the following procedure, where at each step the length of the considered time window is kept fixed at four months, on the basis of the findings reported in Section \ref{subsection - length error}:
\begin{enumerate}[label=\textbf{A.\arabic*}]
\item \label{step 1 stability} apply the calibration algorithm with a time window starting at day $d_0$;
\item \label{step 2 stability} perform the calibration over a time window starting at day $d_0 + 1$, using as initial guess the parameter values $\theta^*$ estimated at step \ref{step 1 stability};
\item repeat the previous step, rolling the time window by one day for 50 consecutive steps.
\end{enumerate}

Table \ref{stat theta 0 varying firs 50} reports the results of this procedure, giving the standard deviation of the calibrated parameters.
This remarkable stability can be partly explained by the procedure employed. Indeed, in line with the parameter recalibration procedure widely adopted in market practice, the initial guess $\theta_0$ for the calibration at iteration $i$ is chosen as the value $\theta^*$ estimated at iteration $i-1$. Since the time window is kept fixed at a length of four months, shifting the time window by one day at each step does not alter significantly the market information, thus explaining the stability of the calibrated parameters.

\begin{table}[ht]
\begin{tabular}{lllllllll}
\hline
{} &      $ a^0 $&  $ \sigma^0$ &   $    a^1$ & $  \sigma^1$ & $      a^2$ & $ \sigma^2$ &  $  \beta^1$ & $   \beta^2$ \\
\hline
avg &  0.371948 &  0.164252 &  0.372120 &  0.159068 &  0.372732 &  0.159813 &  0.481433 &  0.882557 \\
std &  0.000004 &  0.000006 &  0.000003 &  0.000006 &  0.000004 &  0.000004 &  0.000002 &  0.000003 \\
\hline
\end{tabular}
\caption{\footnotesize{Average and standard deviation of the parameters calibrated over $[d_0,d_0 + 50$ days].}}
\label{stat theta 0 varying firs 50}
\end{table}

\bibliographystyle{alpha}
\bibliography{BibMC}

\end{document}